\documentclass[pra,twocolumn,amsmath,amssymb,superscriptaddress]{revtex4-2}
\usepackage[utf8]{inputenc}
\usepackage[dvipsnames]{xcolor}
\usepackage{amsfonts}
\usepackage{amsmath}
\usepackage{amssymb}
\usepackage{amsthm}
\usepackage{svg}
\usepackage{quantikz}
\usepackage{xfrac} 
\usepackage{graphicx}
\usepackage{braket}
\usepackage{lipsum}
\usepackage{ORCIDinREVTeX}

\usepackage{bm}
\usepackage[makeroom]{cancel}

\definecolor{myred}{rgb}{0.8500,0.3250,0.0980}
\definecolor{mygreen}{rgb}{0.4660,0.6740,0.1880}
\usepackage{tikz}
\usetikzlibrary{positioning}

\usepackage{comment}

\newtheorem{theorem}{Theorem}

\DeclareRobustCommand{\rchi}{{\mathpalette\irchi\relax}}
\newcommand{\irchi}[2]{\raisebox{\depth}{$#1\chi$}} 

\usepackage[colorlinks,citecolor=blue,linkcolor=blue]{hyperref}
\usepackage[capitalise,nameinlink]{cleveref}
\crefname{section}{Sec.}{Secs.}
\crefname{table}{Tab.}{Tabs.}
\crefname{figure}{Fig.}{Figs.}
\crefname{definition}{Def.}{Defs.}
\crefname{lema}{Lem.}{Lems.}
\crefname{theorem}{Thm.}{Thms.}
\crefname{corollary}{Cor.}{Cors.}

\begin{document}

\title{Efficient Calculation of Equilibrium Correlation Functions}

\author{Yizhi Shen}
\orcid{0000-0002-4160-5482}
\email{yizhis@lbl.gov}
\affiliation{Applied Mathematics and Computational Research Division,
            Lawrence Berkeley National Laboratory,
            Berkeley, CA 94720, USA}

\author{Roel~Van~Beeumen}
\orcid{0000-0003-2276-1153}
\affiliation{Applied Mathematics and Computational Research Division,
            Lawrence Berkeley National Laboratory,
            Berkeley, CA 94720, USA}

\author{Wibe A. de Jong}
\orcid{0000-0002-7114-8315}
\affiliation{Applied Mathematics and Computational Research Division,
            Lawrence Berkeley National Laboratory,
            Berkeley, CA 94720, USA}

\author{Efekan K\"okc\"u}
\orcid{0000-0002-7323-7274}
\email{efekan.kokcu@ucf.edu}
\affiliation{University of Central Florida, Orlando, FL 32816, USA}
\affiliation{Applied Mathematics and Computational Research Division,
            Lawrence Berkeley National Laboratory,
            Berkeley, CA 94720, USA}

\date{\today}

\begin{abstract}
A majority of measurable dynamic quantities of a physical system is described as an equilibrium correlation function of the form $G_{AB}(t-t’) = \left<A(t)B(t’)\right>$, where $t’$ represents the impact time and $t$ represents the response time. Conventional quantum algorithms to compute such quantities via Hamiltonian simulation such as the Hadamard test, variational methods, or linear-response-based algorithms share one feature in common: each time point $t-t’$ is calculated by a different quantum circuit, leading to at least $O(N_t)$ total shots, and at least $O(N_t^2)$ total quantum gates. Here, we provide a different approach: utilizing $O(\log N_t)$ ancilla qubits to store time as a quantum variable, and effectively parallelize the computation of an equilibrium correlation function by requiring only a single quantum circuit for all time points. We show that this circuit needs to be run only $O(\log N_t)$ times, which in total requires $O(N_t \log N_t)$ quantum resources in the large system size limits. We prove that this scaling is optimal, and demonstrate our algorithm calculating the Green’s function of a one-dimensional spinless Hubbard model.      
\end{abstract}

\maketitle

\section{Introduction}\label{sec:intro}%

Almost everything measured about an interacting
quantum system arrives as a correlation function. A scattering or spectroscopy experiment disturbs the system with one operator and monitors the response of
another, leading to expectation values of the form:
\begin{align}
    G_{AB}(t,t') = \bra{\psi} A(t)^\dagger B(t') \ket{\psi},
\end{align}
where $A(t) = e^{i H t} A e^{- i H t}$ is the operator $A$ evolved in time under the system Hamiltonian $H$ (the Heisenberg picture) and $\ket{\psi}$ is the state being probed. Inelastic neutron scattering measures the spin-spin correlation function, whose Fourier transform in space and time gives the dynamical structure factor~\cite{vanHove1954,Squires2012}. Optical conductivity is a current-current correlation function~\cite{Kubo1957}. Angle-resolved photoemission measures the
single-particle Green’s function, the correlation function of a particle’s creation and annihilation operators, whose Fourier transform gives the energies and lifetimes of electronic
excitations~\cite{Damascelli2003,Mahan}. The same objects organize the
theory. Linear response and the fluctuation-dissipation relation express every susceptibility and transport coefficient through equilibrium correlation functions~\cite{Kubo1957,Callen1951,Bruus2004,Stefanucci2013}, and the Green’s function is the quantity passed back and forth between the impurity solver and the lattice
in dynamical mean-field theory and related embedding methods~\cite{Georges1996,Kotliar2006,Bauer16,Rungger2019}. A method that computes correlation functions does more than reproduce a spectrum. It delivers the input that the rest of many-body physics consumes.

Quantum computers offer a natural place to compute them. The step that makes the classical calculation hard, evolving a many-body state forward in real time, is the step a quantum computer performs efficiently~\cite{Feynman1982,Lloyd1996,Childs2021,Low2019}. A quantum simulation of a correlation function has three parts: prepare the state $\ket{\psi}$, evolve it, and then measure.
State preparation is a problem in its own right with specific algorithms~\cite{Lin2020,Motta2020,Dalzell2023}, and throughout this work we take the state as given. The measurement is where the cost of correlation functions hides. Unlike an energy or magnetization, $G_{AB}(t,t')$ is not the expectation value of a single Hermitian observable at a single time. It is the overlap between two histories of the same state, one
in which $B$ acted at time $t'$ and one in which $A$ acted at time $t$, which can be extracted through interference. The standard tool for this is the Hadamard
test~\cite{Ortiz2001,Somma2002}. An ancilla qubit is prepared in a superposition,
$A$ and $B$ are applied to the system only on the branch where the ancilla reads $\ket{1}$, and measuring the
ancilla along $X$ or $Y$ returns the real or imaginary part
of $G_{AB}$. Most quantum algorithms for Green’s functions and spectra are built on this circuit or on close relatives of it: proposals and hardware demonstrations for dynamical
mean-field theory~\cite{Kreula2016,Keen2020,Steckmann2023,hogan2026efficient}, for neutron scattering and
magnon spectra~\cite{chiesa2019quantum,francis2020quantum}, for molecular Green’s functions~\cite{endo2020calculation,libbi2022}, and for for general multi-time correlation functions~\cite{Pedernales2014}. 

Alternatives exist, and each trades the resources in a different way. Linear-response methods dispense with the ancilla by simulating the experiment itself, where a weak time dependent
field coupled to $B$ drives the system and the
change in $A$ is recorded~\cite{Baez2020,Kokcu2024,Piccinelli2026circuit}. Frequency-domain methods retrieve $G_{AB}(\omega)$ directly, by solving the linear
equations that define the resolvent $(\omega - H)^{-1}$ on quantum computer~\cite{Tong21}, filtering the state in energy~\cite{Roggero20,Lu2021,Keen2021}, or employing quantum phase estimation~\cite{Kosugi20linear,Sels2021}. Variational and subspace methods trade circuit depth for
classical optimization~\cite{Chen2021,Huang2022jpcl,Jamet2022,Jensen2023}, while spectral-estimation methods reconstruct energy spectra from controlled time evolution via suitable signal processing~\cite{Somma2019,Lin2022,Dong2022,Stroeks2022}. All these methods share a bookkeeping fact that is often taken as given: one circuit addresses one point. 
For a stationary state, such as an eigenstate or thermal state of a time independent Hamiltonian
$H$, the correlation function depends only on the time difference, \emph{i.e.}, $G_{AB}(t,t') = G_{AB}(t - t')$, and is called an equilibrium correlation function. Stationarity hence reduces a two-time grid to a one-time grid, after which the remaining $N_t$ values are measured one at a time.
The standard bookkeeping of one circuit per one time/frequency point then requires $N_t$ distinct circuits to measure $G_{AB}(t-t')$ for $N_t$ points.
Each circuit must be repeated
$O(\log(1/\delta')/\epsilon^2)$ times to reach a target precision $\epsilon$ at confidence $1-\delta'$, and the full procedure
therefore costs $O(N_t \log(N_t/\delta)/\epsilon^2)$ circuit executions at
confidence $1 - \delta \le (1-\delta')^{N_t}$, with the same bill for every additional
operator pair.

\begin{figure*}[t!]
    \centering
    \begin{quantikz}[
  execute at end picture={
    \draw[decorate,decoration={brace,amplitude=6pt},line width=1pt]
      ([xshift=4pt]\tikzcdmatrixname-1-6.north east) -- ([xshift=4pt]\tikzcdmatrixname-1-6.south east)
      node[midway,xshift=10pt,anchor=west,font=\large]{$\langle\psi|A(t_j)^\dagger B|\psi\rangle$};}]
        \lstick{$|0\rangle$} & \gate{H} & \ctrl{1} & \qw & \octrl{1} & \meter{X\text{ or }Y} \\
        \lstick{$|\psi\rangle$} & \qw & \gate{B} & \gate{U^{j}} & \gate{A} & \qw
    \end{quantikz}\\[6pt]
    (a)~\\[10pt]
    \begin{quantikz}[column sep=0.18cm,
  execute at end picture={
    \draw[decorate,decoration={brace,amplitude=6pt},line width=1pt]
      ([xshift=4pt]\tikzcdmatrixname-1-13.north east) -- ([xshift=4pt]\tikzcdmatrixname-2-13.south east)
      node[midway,xshift=10pt,anchor=west,font=\large]{$\langle\psi|A(t)^\dagger A(t')|\psi\rangle$};}]
        \lstick{$|0\rangle$} & \qw & \qw & \qw & \qw & \qw & \qw & \hspace{0.04in}\cdots\hspace{0.04in} & \qw & \qw & \qw & \gate[wires=2]{\mathrm{QFT}} & \meter{} \\
        \lstick{$|0\rangle^{\otimes b}$} & \qw & \gate{H^{\otimes b}} & \gate{0}\vqw{1} & \qw & \gate{1}\vqw{1} & \qw & \hspace{0.04in}\cdots\hspace{0.04in} & \qw & \gate{N_t{-}1}\vqw{1} & \qw & \qw & \meter{} \\
        \lstick{$|\psi\rangle$} & \qw & \qw & \gate{A} & \gate{U} & \gate{A} & \gate{U} & \hspace{0.04in}\cdots\hspace{0.04in} & \gate{U} & \gate{A} & \qw & \qw & \qw
    \end{quantikz}\\[6pt]
    (b)\\[10pt]
    \begin{quantikz}[column sep=0.18cm,
  execute at end picture={
    \draw[decorate,decoration={brace,amplitude=6pt},line width=1pt]
      ([xshift=4pt]\tikzcdmatrixname-1-15.north east) -- ([xshift=4pt]\tikzcdmatrixname-3-15.south east)
      node[midway,xshift=10pt,anchor=west,font=\large]{$\langle\psi|A(t)^\dagger B(t')|\psi\rangle$};}]
        \lstick{$|0\rangle$} & \qw & \qw & \qw & \qw & \qw & \qw & \qw & \qw & \hspace{0.04in}\cdots\hspace{0.04in} & \qw & \qw & \qw & \gate[wires=2]{\mathrm{QFT}} & \meter{} \\
        \lstick{$|0\rangle^{\otimes b}$} & \qw & \gate{H^{\otimes b}} & \gate{0}\vqw{2} & \gate{0}\vqw{2} & \qw & \gate{1}\vqw{2} & \gate{1}\vqw{2} & \qw & \hspace{0.04in}\cdots\hspace{0.04in} & \qw & \gate{N_t{-}1}\vqw{2} & \gate{N_t{-}1}\vqw{2} & \qw & \meter{} \\
        \lstick{$|0\rangle$} & \qw & \gate{H} & \octrl{1} & \ctrl{1} & \qw & \octrl{1} & \ctrl{1} & \qw & \hspace{0.04in}\cdots\hspace{0.04in} & \qw & \octrl{1} & \ctrl{1} & \qw & \meter{X\text{ or }Y} \\
        \lstick{$|\psi\rangle$} & \qw & \qw & \gate{A} & \gate{B} & \gate{U} & \gate{A} & \gate{B} & \gate{U} & \hspace{0.04in}\cdots\hspace{0.04in} & \gate{U} & \gate{A} & \gate{B} & \qw & \qw
    \end{quantikz}\\[6pt]
    (c)
    \caption{Comparison of circuits for the equilibrium two-point correlation function $G_{AB}(t_j)=\bra{\psi}A(t_j)^\dagger B\ket{\psi}$ on the time grid $t_j=j\Delta t$, with $U=e^{-i\Delta t H}$. (a) Hadamard test circuit. The ancilla selects which operator acts, \emph{i.e.}, $A$ on the $\ket{0}$ branch (open circle) and $B$ on the $\ket{1}$ branch (filled circle), and its $X$ or $Y$ expectation value recovers the real or imaginary part of $G_{AB}(t_j)$ for the single time $t_j$ fixed by the circuit. (b) Clock circuit for the autocorrelation function $G_{AA}(t_j)=\bra{\psi}A(t_j)^\dagger A\ket{\psi}$. A clock register of $b=\log_2 N_t$ qubits in an equal superposition sets the insertion time: the boxed value $r$ on the clock wire conditions the insertion of $A$ on the clock value $r$, and the system evolves through uncontrolled time steps $U$. A quantum Fourier transform (QFT) over the clock and one padding qubit, followed by their measurement, yields samples from which $G_{AA}(t_j)$ is estimated simultaneously for all $|j| \leq N_t/2$. (c) Clock circuit for the cross-correlation function $G_{AB}(t_j)$ of two different operators $A$ and $B$. A selector ancilla chooses $A$ or $B$ at every clocked insertion as in (a). Its $X$ or $Y$ measurement, combined with the QFT outcome, estimates the real or imaginary part of $G_{AB}(t_j)$ simultaneously for all $|j| \leq N_t/2$.}
    \label{fig:circ_comp_AB}
\end{figure*}
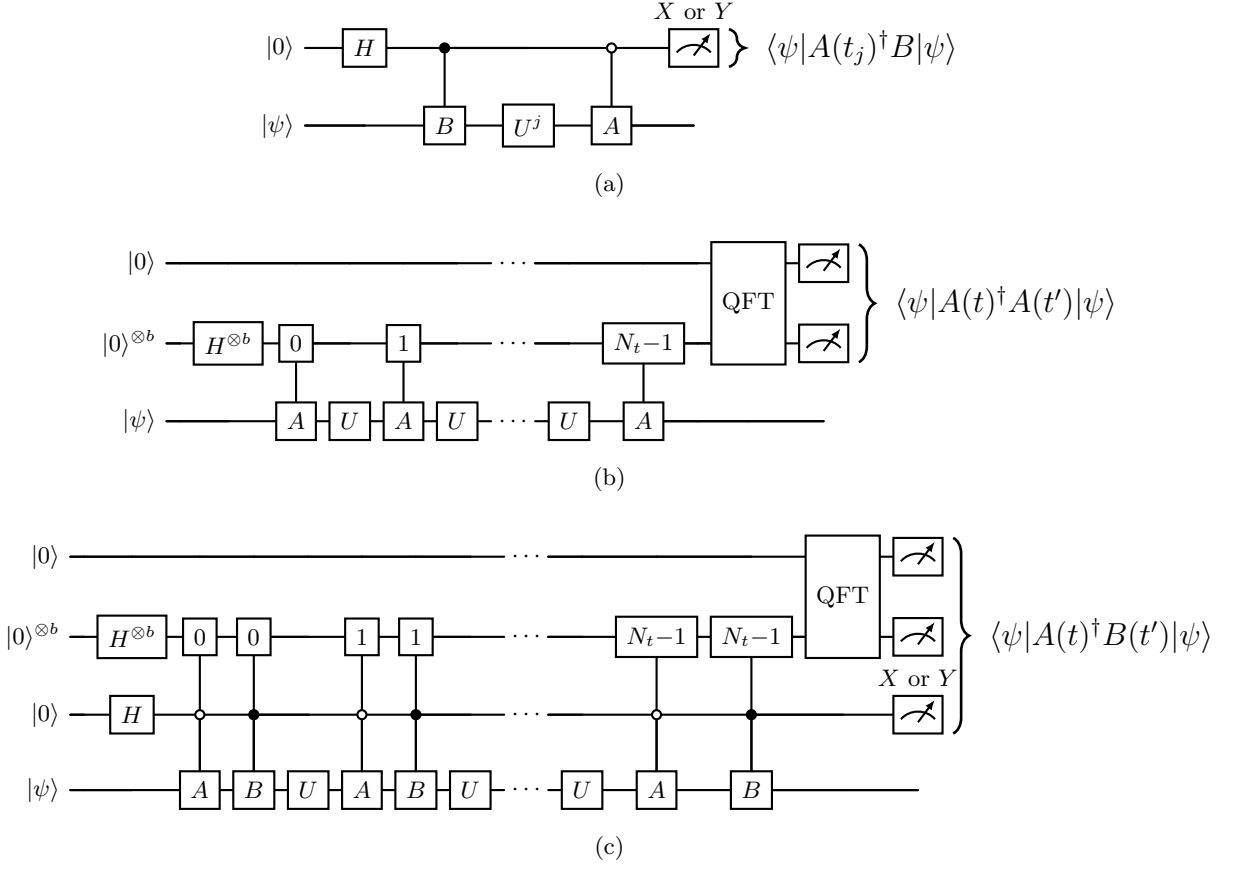

In this work, we remove such bookkeeping constraint for (equilibrium) correlation functions by promoting the insertion time to a quantum variable. A clock register of $\log_2 N_t$ qubits is prepared in a superposition over $N_t$ time points. The system then evolves freely under its Hamiltonian, with no control on the evolution, while the operator $A$ is applied at each time step on the branch selected by the clock. This produces a superposition of $N_t$ histories, each in which $A$ acts at a different instant and is tagged accordingly by the clock. For stationary states, pairwise interference among these histories is determined by the difference between their insertion times, allowing the clock register to record the correlation function at all $N_t$ lags simultaneously. A quantum Fourier transform (QFT) then converts this interference into a frequency distribution over the transitions excited by $A$, broadened by a known window.
Each measured frequency provides an unbiased sample for all $N_t$ lags, and Hoeffding’s inequality~\cite{Hoeffding1963} ensures that they can be estimated to precision $\epsilon$ with $O(\log (N_t/\delta)/\epsilon^2)$ shots.
The number of circuit executions therefore drops from linear to logarithmic in the number of time points. The same
construction with one additional control qubit handles two different probe
operators $A$ and $B$. Furthermore, the argument does not assume a pure state, so thermal state correlations can be covered on the same footing.

Several aspects of our construction are particularly favorable for quantum hardware. First, the Hamiltonian simulation part is never controlled. Only the insertion of $A$ depends on the clock. When $A$ is Pauli, its insertion involves a multi-controlled gate with a non-Clifford cost scaling as $O(\log N_t)$~\cite{Gidney2018,Nie2024,Khattar2025}, which dominates the fault-tolerant overhead. This contrasts with quantum phase estimation~\cite{Kitaev1995,Cleve1998}, where the evolution itself is controlled by an ancilla register throughout the dynamics. Second, because the QFT immediately proceeds measurement, it can be implemented semiclassically using single-qubit rotations and classical feedforward, with no additional entangling gates~\cite{Griffiths1996}. The system register need not be measured. Third, the clock register also realizes a history state in the sense of Feynman and Kitaev~\cite{Feynman1986,Kitaev2002}, but here it serves as a measurement device rather than an encoding of a computation. Measuring it in the Fourier basis resolves the transition frequencies of $H$, \emph{i.e.}, the energy differences coupled by $A$, while leaving the underlying dynamics uncontrolled. Finally, the separation between free Hamiltonian evolution and clock-conditioned Pauli insertions also makes the protocol naturally compatible with analog simulators supplemented by a small digital register~\cite{Bloch2012,Monroe2021,Browaeys2020}.

Our contributions are as follows. We present the algorithm and prove the core identity underlying it in \cref{sec:algorithm}. The reduced clock state is the correlation function, its Fourier statistics are the windowed spectral function, and
both statements are exact for any stationary state. We
establish an $O( \log(N_t/\delta)/\epsilon^2)$ sample complexity for all
lags simultaneously and show that no measurement on independent copies of the clock state, however elaborate,
can improve on it. In \cref{sec:compare_hadamard}, we compare the resulting resource requirements with those of the Hadamard test, including the costs from Trotterization- and qubitization-based time evolution. In \cref{sec:examples}, we demonstrate our protocol numerically for an interacting fermionic system and characterize its finite-shot behavior. We 
summarize and conclude in \cref{sec:discussion}.

\section{The Algorithm}
\label{sec:algorithm}

We take the Hamiltonian $H$ be time independent, and the state of the system to be stationary. A pure state $\ket{\psi}$ satisfies $H\ket{\psi}=E\ket{\psi}$, whereas a mixed state $\rho$ satisfies $[H,\rho]=0$. For clarity, we focus on pure states and note where mixed states differ. We denote the one step unitary by $U = e^{-i\Delta t H}$ and define an $N_t$-point grid, $t_r = r \Delta t$, with
\begin{align}
   r = 0, 1, 2, \dots, N_t-1.
\end{align}
On this grid, the Heisenberg-evolved probe is $A(t_r)=U^{-r}AU^r$. We will assume that $A$ is unitary, for instance a Pauli string. Non-unitary probes, including fermionic operators, can be expanded in Pauli basis and treated by the two-operator construction of \cref{subsec:two_ops}.

Let us first consider the equilibrium autocorrelation function:
\begin{align}
    G_{AA}(t_s - t_r) = \bra{\psi} A(t_s)^\dagger A(t_r) \ket{\psi},
\end{align}
where unitarity of $A$ implies $G_{AA}(0) = 1$, $\lvert G_{AA}(t) \rvert \leq 1$, and $G_{AA}(-t) = G_{AA}(t)^{\ast}$. The autocorrelation function also admits a spectral representation useful throughout. Expanding $A\ket{\psi}=\sum_n a_n\ket{E_n}$ in the eigenstates of $H$, we obtain
\begin{align}
G_{AA}(t)
= \sum_n \lvert a_n \rvert^2 e^{-i(E_n-E)t}
= \int d\nu \, S(\nu)e^{-i\nu t},
\label{eq:spectral}
\end{align}
where $S(\nu)=\sum_n \lvert a_n \rvert ^2\,\delta(\nu - (E_n-E) )$ is the spectral function of $A$ that assigns weight $\lvert a_n \rvert ^2$ to each transition frequency $E_n - E$ accessible from $\ket{\psi}$.
In particular, $S(\nu)$ is nonnegative and, since $G_{AA}(0) = 1$, normalized. Thus, $G_{AA}$ is the Fourier transform of a probability distribution over transition frequencies.
For a stationary mixed state $\rho = \sum_n p_n \ket{\psi_n} \bra{\psi_n}$, the same representation holds with frequencies $E_m-E_n$, now weighted by $p_n|A_{mn}|^2$.

\subsection{Two-point correlations in a clock register}
\label{subsec:clock_register}

We introduce a clock register of $b = \log_2 N_t$ ancillas and prepare it in a uniform superposition with a single layer of Hadamard gates. This gives the initial state,
\begin{align}
    \frac{1}{\sqrt{N_t}} \sum_{r=0}^{N_t-1} \ket{r} \otimes \ket{\psi}.
\end{align}
Let $C_r(A)=\ket{r}\bra{r}\otimes A+(I-\ket{r}\bra{r})\otimes I$ denote the insertion of $A$ conditioned on the clock state $\ket{r}$. The circuit, shown in \cref{fig:circ_comp_AB}(b) with the selector qubit omitted, is the interleaving sequence,
\begin{align}
    C_{N_t-1}(A) U \cdots U C_1(A) U C_0(A),
    \label{eq:sweep}
\end{align}
read from right to left. Specifically, on the clock branch $\ket{r}$, $A$ is inserted after $r$ applications of $U$. Because each $U$ acts unconditionally on the system, the Hamiltonian evolution remains uncontrolled; the clock affects only the insertions. The resulting state is
\begin{align}
\begin{split}
    \ket{\Psi} =&\frac{1}{\sqrt{N_t}} \sum_{r=0}^{N_t-1} \ket{r} \otimes U^{N_t-r-1} A U^r \ket{\psi} 
    \\
    = &\frac{1}{\sqrt{N_t}} \sum_{r=0}^{N_t-1} \ket{r} \otimes U^{N_t-1} A(t_r) \ket{\psi}.
\end{split}
\label{eq:history}
\end{align}
The factor $U^{N_t-1}$ is common to all the branches and does not affect their overlaps. Up to this common evolution, the state is a superposition of $N_t$ histories, one for each possible insertion time of $A$, with the clock labeling the histories. This is a history state in the sense of Feynman and Kitaev \cite{Feynman1986,Kitaev2002}, used here as a measuring instrument. Tracing out the system leaves the clock in the state,
\begin{align}
\begin{split}
    \rho_{\mathrm{Clk}} =& \frac{1}{N_t} \sum_{r,s = 0}^{N_t-1} \ket{r}  \bra{s} \:\cdot \: \bra{\psi}A(t_s)^\dagger A(t_r) \ket{\psi}, \\
    =& \frac{1}{N_t} \sum_{r,s = 0}^{N_t-1} \ket{r}  \bra{s} \:\cdot \: G_{AA}(t_s - t_r),
\end{split}
\label{eq:rhoclk}
\end{align}
Thus, all grid values of the correlation function appear as matrix elements of the $N_t\times N_t$ clock state $\rho_{\mathrm{Clk}}$, stored on $\log_2N_t$ qubits by a single circuit. Stationarity makes $\rho_{\mathrm{Clk}}$ Toeplitz, with each diagonal corresponding to a fixed time lag, while its positivity reflects the positive-definite structure of $G_{AA}$. For a stationary mixed state, the same circuit gives \cref{eq:rhoclk} with $G_{AA}(t)=\mathrm{Tr}[\rho\,A(t)^\dagger A]$.

In the energy basis the clock state reads off the spectral function explicitly. Substituting the expansion of $A\ket{\psi}$ into \cref{eq:history} and using $H\ket{\psi}=E\ket{\psi}$,
\begin{align}
\begin{split}
    \ket{\Psi}&=\sum_n a_n\,e^{-iE_n(N_t-1)\Delta t} \ket{\rchi(E_n-E)}\otimes\ket{E_n},\\
    \ket{\rchi(\nu)}&=\frac{1}{\sqrt{N_t}}\sum_{r=0}^{N_t-1}e^{i\nu r\Delta t}\ket{r},
\end{split}
    \label{eq:phasestate}
\end{align}
such that $\rho_{\mathrm{Clk}}=\sum_n \lvert a_n\vert^2 \ket{\rchi(E_n-E)}\bra{\rchi(E_n-E)}$. The spectral function is therefore represented in the clock as a mixture of phase states $\ket{\rchi}$. This parallels the textbook quantum phase estimation. Here, nonetheless, the phases encode transition frequencies $E - E_n$ rather than absolute energies $E_n$, and are acquired under the same control-free time evolution on every clock branch.

\subsection{Fourier readout}
\label{subsec:readout}

The clock state \cref{eq:rhoclk} contains $G_{AA}$ in its off-diagonal entries, which cannot be read out in the computational basis. A Fourier transform converts this information into directly measurable statistics. If $\rho_{\mathrm{Clk}}$ were circulant, \emph{i.e.}, Toeplitz with periodic wrap-around, a quantum Fourier transform (QFT) on the clock state would diagonalize it exactly. However, the time difference $t_s-t_r$ spans $2N_t-1$ distinct lags, while an $N_t$-dimensional Fourier transform can only resolve $N_t$ frequency modes, causing the positive and negative lags to alias. To address this ambiguity, we employ a single padding qubit, initialized in $\ket{0}$ and left untouched until the QFT, thereby embedding the clock state in a $2N_t$-dimensional register:
\begin{align}
    \tilde{\rho}_{\mathrm{Clk}} = \begin{pmatrix} \rho_{\mathrm{Clk}} & 0 \\ 0 & 0
    \end{pmatrix}.
\end{align}
Crucially, the extra qubit eliminates the aliasing without changing the system dynamics. No additional time steps or controlled Hamiltonian evolution are required. 

We then apply the QFT to the $1+\log N_t$ ancillas and measure them in the computational basis. An outcome $k \in \{0, \ldots, 2N_t - 1\}$ corresponds to the frequency $\omega_k = k \pi / (N_t \Delta t)$, with frequency spacing $\pi / (N_t \Delta t)$~\footnote{Here we adopt ${\rm QFT}\ket{r} = (2N_t)^{-1/2} \sum_{k} e^{-i\omega_k t_r} \ket{k} $. The opposite sign convention exchanges $\omega$ with $-\omega$.}. The probability of sampling $\omega \in \{ 0, \ldots,  (2N_t - 1)\pi/(N_t \Delta t) \}$ is the diagonal element of the Fourier-transformed clock state at that frequency. Grouping the terms in \cref{eq:rhoclk} by the lag $j=s-r$, which occurs for $N_t-\lvert j \rvert$ pairs $(r,s)$, gives
\begin{align}
    \tilde{\rho}_{\omega \omega} = \frac{1}{2N_t} \sum_{j=-(N_t-1)}^{N_t-1} \left(1 - \frac{|j|}{N_t} \right) e^{i t_j\omega} G_{AA}(t_j). 
\label{eq:pomega}
\end{align}
for $\omega \in \pi \mathbb{Z} / N_t \Delta t$. Equation \eqref{eq:pomega} has two complementary interpretations. Substituting the spectral representation \cref{eq:spectral} and evaluating the finite sum yields
\begin{align}
\begin{split}
\tilde{\rho}_{\omega\omega}
&=\frac{1}{2N_t}\int d\nu \,S(\nu)
F_{N_t}((\omega-\nu)\Delta t),\\
F_{N_t}(\theta)
&=\frac{1}{N_t}\frac{\sin^2(N_t \theta/2)}
{\sin^2(\theta/2)},
\end{split}
\label{eq:fejer}
\end{align}
where $F_{N_t}$ is the Fej{\'e}r kernel. Thus, the Fourier-outcome histogram is the spectral function smoothed by a window of width $O(2\pi/(N_t\Delta t))$, set by the total evolution time. The clock realizes this finite-time-broadened frequency distribution, without first reconstructing the correlation function at individual lags.

Additionally, \cref{eq:pomega} is a finite Fourier series whose coefficients are the correlation values multiplied by the triangular factor $1 - \lvert j \rvert/N_t$. Inverting this series over the $2N_t$ Fourier modes gives, for every $\lvert j \rvert \le N_t - 1$,
\begin{align}
\begin{split}
    G(t_j) =& \frac{N_t}{N_t - |j|} \sum_\omega e^{-it_j \omega} \tilde{\rho}_{\omega \omega} \\ 
    =& \frac{N_t}{N_t - |j|} \mathbb{E}_\omega \left[ e^{-it_j \omega}\right],
\end{split}
\label{eq:lag}
\end{align}
where $\mathbb{E}_\omega$ is the expectation value under the probability distribution $\tilde{\rho}_{\omega \omega}$. Over the doubled frequency grid, phase factors associated with all lags $\lvert j\rvert\leq N_t-1$ are mutually orthogonal so that each $G_{AA}(t_j)$ can be recovered exactly from this distribution.
The factor $N_t/(N_t-\lvert j\rvert)$ corrects for the smaller number of $(r,s)$ pairs at larger lags and increases as $\lvert j\rvert$ approaches its maximum value of $N_t-1$. Keeping this factor $O(1)$ prevents the statistical error [c.f. RHS of \cref{eq:lag}] from growing with $N_t$. We thus restrict to $\lvert j\rvert\leq N_t/2$ for which $N_t/(N_t-\lvert j\rvert) \leq 2$. The restricted range still contains all $N_t+1$ lags $j=-N_t/2,\ldots,N_t/2$, obtained from the same stream of Fourier measurements.

\subsection{Optimal sampling of all time lags}
\label{subsec:samples}
Equation \eqref{eq:lag} also provides a natural estimator for the correlation function. After $M$ circuit executions yielding outcomes $\omega^{(1)},\ldots,\omega^{(M)}$, we define
\begin{align}
    \hat{G}_{AA}(t_j)
    =\frac{N_t}{N_t-|j|}
    \frac{1}{M}\sum_{\mu = 1}^{M}e^{-it_j\omega^{(\mu)}}.
    \label{eq:estimator}
\end{align}
The same set of outcomes can be shared across all lags, and each estimate is unbiased by our construction. Since $\lvert e^{-it_j\omega}\rvert=1$, Hoeffding's inequality~\cite{Hoeffding1963} implies that, for either the real or imaginary part of the sample mean, the probability of a deviation greater than $\eta/\sqrt{2}$ is at most $2e^{-M\eta^2/4}$. Hence, the deviation from $\mathbb{E}_\omega[e^{-it_j\omega}]$ exceeds $\eta$ with probability at most $4e^{-M\eta^2/4}$.  
This corresponds to $\epsilon_j = N_t/(N_t-\lvert j\rvert) \eta$ accuracy for $G_{AA}(t_j)$ computation.
If we only consider $N_t+1$ lags with $\lvert j\rvert\le N_t/2$, we obtain $\epsilon_j \le 2 \eta$.
Applying a union bound over these $N_t+1$ lags 
gives the following result:

\begin{theorem}[Sample complexity]\label{thm:samples}
Let $0<\epsilon,\delta<1$. With
\begin{align}
    M=\Big\lceil\frac{16}{\epsilon^2}\log\frac{4(N_t+1)}{\delta}\Big\rceil
    \label{eq:M}
\end{align}
executions of the circuit, the estimates \cref{eq:estimator} satisfy $|\hat{G}_{AA}(t_j)-G_{AA}(t_j)|\le\epsilon$ simultaneously for all $|j|\le N_t/2$ with probability at least $1-\delta$.
\end{theorem}

The number of shots is therefore $O(\log(N_t/\delta)/\epsilon^2)$, only logarithmic in the number of time points. By contrast, the Hadamard test needs a separate circuit for each lag and $O(\log(N_t/\delta)/\epsilon^2)$ repetitions per circuit to reach the same guarantee, resulting in an additional factor of $N_t$ in the total shot count. We make this comparison at equal total evolution time in \cref{sec:compare_hadamard}. The same measurement record can also be post-processed in different ways. The $M$ outcomes yield the time-domain estimates above, the frequency distribution in \cref{eq:fejer}, and linear functionals of 
$G_{AA}$ needed by downstream calculations, for example exponentially damped or imaginary-time correlators (the latter used in embedding methods). These quantities can all be extracted at the resolution fixed by the finite time window, with no additional circuit executions.

The logarithmic dependence of the sample complexity on $N_t$ is not an artifact of the QFT readout. In fact, any unrestricted protocol that can store independent copies of the clock state and measure them collectively must use $M = \Omega(\log N_t/\epsilon^2)$ copies to estimate all lags to additive error $\epsilon$. It follows that our approach attains the optimal $N_t$-scaling, up to constant factors, in the sample-access model.

\begin{theorem}[Copy lower bound]\label{thm:lower}
Let $\epsilon\le1/8$. Any protocol that receives $M$ independent copies of the clock state $\rho_{\mathrm{Clk}}$, performs an arbitrary collective measurement on them, and outputs estimates of $G_{AA}(t_j)$ for all $|j|\le N_t/2$ with additive error at most $\epsilon$ and success probability at least $2/3$ on every stationary instance must use
\begin{align}
    M\ge\frac{\log_2(N_t/2)-3/2}{96\,\epsilon^2}
    \label{eq:lower}
\end{align}
copies.
\end{theorem}

\begin{proof}
Set $a=4\epsilon\le 1/2$. For each $j_0\in\{1,\ldots,N_t/2\}$, consider an autocorrelation function with
$$ G(0)=1, \qquad G(\pm t_{j_0})=a, $$
and $G(t_j) = 0$ for other lags $j \neq j_0$. This sequence is the Fourier series of the density
$[1+2a\cos(j_0\theta)]/2\pi$
and, thus, defines a valid stationary instance: the spectral function with the following amplitudes and frequencies
\begin{align}
   a_n=\sqrt{ \frac{1+2a\cos(2\pi j_0 n/2N_t)}{2N_t}}, \qquad \nu_n = \frac{n \pi}{N_t\Delta t}, \nonumber
\end{align}
indeed satisfy $\sum_{n=0}^{2N_t - 1} \lvert a_n \rvert^2 e^{-i\nu_n t_j}=\delta_{j,0}+a (\delta_{j,j_0}+\delta_{j,-j_0})$ for each $|j|\le N_t-1$. One may realize it, through \cref{eq:spectral}, by choosing $H$ with levels $E + \nu_n$, $\ket{\psi}$ to be an eigenstate at energy $E$, and $A$ that sends $\ket{\psi}$ to the corresponding superposition $\sum_n a_n \ket{E_n}$.

By \cref{eq:rhoclk}, the resulting clock states are
\begin{align}
    \rho_{{\rm Clk},j_0} = \frac{I+a (\mathbb{T}_{j_0} + \mathbb{T}_{j_0}^\dagger )}{N_t}, \qquad \rho_{{\rm Clk},0} = \frac{I}{N_t}, \nonumber
\end{align}
where $\mathbb{T}_{j_0}$ is the truncated shift on the $j_0^{\rm th}$ superdiagonal, and $\rho_{{\rm Clk}, 0}$ denotes the case $a=0$. The operator $\mathbb{T}_{j_0} + \mathbb{T}_{j_0}^\dagger$ is traceless, has eigenvalues in $[-2,2]$, and satisfies
\begin{align}
    \operatorname{Tr} \big[ (\mathbb{T}_{j_0} + \mathbb{T}_{j_0}^\dagger )^2 \big] = 2(N_t - j_0). \nonumber
\end{align}
Therefore $\rho_{{\rm Clk}, j_0 } \ge 0$, and using $\log(1+z)\le z$,
\begin{align}
    D(\rho_{{\rm Clk}, j_0}\Vert\rho_{{\rm Clk}, 0} ) \le \frac{a^2}{N_t} \operatorname{Tr} \big[ (\mathbb{T}_{j_0} + \mathbb{T}_{j_0}^\dagger )^2 \big] \le 2a^2, \nonumber
\end{align}
where $D$ is the quantum relative entropy and $S$ the von Neumann entropy.

Relative entropy is additive under tensor products, so the Holevo information of a uniform ensemble $\{\rho_{{\rm Clk}, j_0}^{\otimes M}\}_{j_0}$ is at most $2Ma^2$. Conversely, an $\epsilon$-accurate estimate of all lags locates $j_0$. The $j=0$ lag carries no information, because $G_{AA}(0) = 1$
for a unitary probe $A$. Among the lags $1 \leq j \leq N_t/2$, the estimate at the marked lag $t_{j_0}$ has magnitude at least $a-\epsilon=3\epsilon$, whereas all other lags have magnitude at most $\epsilon$. Hence, the estimation problem can distinguish among $N_t/2$ equiprobable choices of $j_0$. By Fano's inequality, achieving this with error probability at most $1/3$ requires at least
$$ \frac{2}{3}\log_2(N_t/2)-1 $$
bits of mutual information. The Holevo bound, however, limits the available information to $2Ma^2/\log 2$ bits. 

Combining these bounds and substituting $a=4\epsilon$ gives the stated result.
\end{proof}

\cref{thm:lower} assumes sample access: the protocol receives copies of the clock state, which it may store and measure jointly, but cannot coherently call the state preparation circuit or its inverse. This access model is natural for analog simulators, thermal ensembles, and irreversible state preparation procedures.
When coherent access to the full preparation circuit and its inverse is available, stronger primitives can be invoked. Amplitude estimation~\cite{Brassard2002,Grinko2021} achieves $1/\epsilon$ scaling for a single lag or spectral band, and shared multi-observable estimation~\cite{Huggins2022} can recover all lags with $\widetilde O(\sqrt{N_t}/\epsilon)$ preparation calls. Neither requires controlled Hamiltonian evolution.

\subsection{Resource requirements}
\label{subsec:resources}

The total quantum resources required for the algorithm will primarily depend on how the evolution $U = e^{-i\Delta t H}$ is implemented. Let $\varepsilon$ denote the target error of the full real-time evolution, distinct from the statistical precision $\epsilon$. Our circuit contains $N_t-1$ evolution segments, so it is sufficient to implement each to accuracy $O(\varepsilon/N_t)$. Below we first summarize the register, insertion, and readout overheads. We then estimate the simulation cost using $d^{\rm th}$-order product formulas and qubitization. We include product formulas here because, for local Hamiltonians at practical accuracies, \emph{e.g.}, $\varepsilon \sim 10^{-3}$, they can outperform methods with better asymptotic scaling~\cite{Childs2018,Childs2021}, whereas qubitization becomes advantageous only at sufficiently small $\varepsilon$~\cite{Low2017,Low2019}.

\emph{Qubits.} The algorithm uses the system register, $\log_2 N_t$ clock qubits, as well as one padding qubit. Correlations between two different operators $A$ and $B$ require an extra selector qubit. These counts exclude ancilla workspace needed by the Hamiltonian simulation routine of choice. Unary iteration, when used for the insertions below, adds a small reusable ancilla workspace.

\emph{Probe insertions.} Each $C_r(A)$ applies $A$ conditioned on the clock state being in $\ket{r}$. For a Pauli operator $A$, this is a multi-controlled Pauli gate with $\log_2 N_t$ control qubits. Standard decompositions require $O(\log N_t)$ Toffoli gates and $O(\log\log N_t)$ depth~\cite{Gidney2018,Nie2024,Khattar2025,Dutta2025}, for which the Toffoli count sets the dominant non-Clifford overhead in a fault-tolerant implementation. Because the clock addresses are visited sequentially as $r=0,1\ldots,N_t-1$, unary iteration can update the address decoding incrementally between successive insertions.~\cite{Babbush2018,Khattar2025}.
This makes the total non-Clifford cost $O(N_t)$ across all $N_t$ insertions, by reusing the workspace of a unary-iteration circuit. Equivalently, the insertion cost is at most $O(\log N_t)$ per time step (with no added ancillary workspace), and $O(1)$ amortized (with shared decoding). These are the only operations that are conditioned on the clock.

\emph{Readout.} Because the QFT immediately precedes measurement, it can be implemented in the semiclassical form of Griffiths and Niu~\cite{Griffiths1996}. In particular, the clock qubits are processed and measured one at a time. After each qubit is measured, its outcome determines the phase rotation applied to the qubits that remain. This replaces the controlled rotations of the standard QFT by single-qubit gates and classical feedforward. No entangling gates are therefore required within the clock register. The system register itself need not be measured.

\emph{Time evolution.} Each circuit execution contains $N_t-1$ applications of $U$, amounting to a total physical evolution time $T \simeq (N_t - 1)\Delta t = O(N_t \Delta t)$. This is comparable to a single Hadamard-test circuit evaluated with the longest lag. Moreover, resolving $G_{AA}(t_{N_t-1})$ necessarily requires histories across a time separation $t_{N_t-1}$, so this evolution time cannot be reduced by the clock construction. The computational cost of reaching $T$, however, depends on the Hamiltonian simulation method.

\subsubsection{Time evolution via Trotterization}
For a product formula of order $d$, implementing one segment $U$ to accuracy $\varepsilon/N_t$ takes
\begin{align}
O \left(
\Delta t^{1+1/d}
\frac{N_t^{1/d}}{\varepsilon^{1/d}} \,
\right)
\end{align}
Trotter steps, up to constants determined by commutator norms of the Hamiltonian terms~\cite{Childs2021}. Summing over the $N_t$ segments gives
\begin{align}
O \left(
\Delta t^{1+1/d}
\frac{N_t^{1+1/d}}{\varepsilon^{1/d}}  \,
\right)
=
O \left(
\frac{T^{1+1/d}}{\varepsilon^{1/d}}  \,
\right),
\label{eq:trotter}
\end{align}
which is the usual cost of simulating total evolution time $T$ to accuracy $\varepsilon$, independent of how $T$ is divided among the $N_t$ segments. Combining this with the shot count in \cref{eq:M}, simultaneously estimating all lags incurs a total Hamiltonian simulation gate complexity
\begin{align} 
O \left(
\frac{T^{1+1/d}}{\varepsilon^{1/d}\epsilon^2}
\log\frac{N_t}{\delta} \right).
\end{align}
At fixed $\Delta t$ and tolerances, the quantum resource scales as $O(N_t^{1+1/d}\log N_t)$, or $O(N_t^{1/d}\log N_t)$ per retained lag.

\subsubsection{Time evolution via qubitization}
With qubitization, implementing $U$ to accuracy $\varepsilon/N_t$ costs
\begin{align}
    O \left( \alpha \Delta t + \frac{\log(N_t / \varepsilon)}{\log\log(N_t / \varepsilon)} \right)
\end{align}
queries to a block encoding of $H$~\cite{Low2017,Low2019}. Here, $\alpha$ is the block encoding normalization, typically given by the sum of the absolute Pauli coefficients. Over all $N_t$ segments, the query complexity is
\begin{align}
    O \left( \alpha T + \frac{N_t \log(N_t / \varepsilon)}{\log\log(N_t / \varepsilon)} \right).
    \label{eq:qubitization}
\end{align}
Again accounting for the shot count in \cref{eq:M}, the total Hamiltonian simulation query complexity is
\begin{align}
O \left[
\left(
\alpha T+
\frac{N_t\log(N_t/\varepsilon)}
{\log\log(N_t/\varepsilon)}
\right)
\frac{\log(N_t/\delta)}{\epsilon^2} \,
\right].
\end{align}
The relative importance of the two terms is influenced by $\alpha$, $T$, $N_t$ and $1/\varepsilon$. Holding $\Delta t$, $\alpha$, and tolerances fixed, $T = O(N_t)$ and the precision-dependent term eventually dominates, entailing $O (N_t(\log N_t)^2 / \log\log N_t )$ quantum resource in total, or $O ( (\log N_t)^2/\log\log N_t )$ per retained lag. Conversely, if $\alpha \Delta t$ dominates the cost per segment, the corresponding scaling becomes $O(N_t\log N_t)$ in total, or $O(\log N_t)$ per time lag.

\subsection{Two-operator construction}
\label{subsec:two_ops}
For two different unitary probes $A$ and $B$, their cross-correlation function $G_{AB}(t)=\bra{\psi}A(t)^\dagger B\ket{\psi}$ needs both families of histories, one in which \(A\) acts and one in which \(B\) acts. A single selector qubit creates this superposition, as shown in \cref{fig:circ_comp_AB}(b). Prepared in \(\ket{+}\), the selector qubit determines which probe is inserted at each clock time \(r\): \(A\) on the \(\ket{0}\) branch and \(B\) on the \(\ket{1}\) branch. The final state is
\begin{align}
    \frac{1}{\sqrt{2N_t}}\sum_{r=0}^{N_t-1}\ket{r}\otimes U^{N_t-1}\Big[&\ket{0}\,A(t_r)
    +\ket{1}\,B(t_r)\Big]\ket{\psi},
    \label{eq:historyAB}
\end{align}
where the reduced state of selector and clock consists of four $N_t\times N_t$ blocks $\rho_{uu'}$, with $u,u'\in\{0,1\}$. The diagonal blocks contain the autocorrelation matrices of $A$ and $B$, each with weight $1/2$. The off-diagonal block $\rho_{10}$ carries matrix elements, $[\rho_{10}]_{rs} = \frac{1}{2N_t} \bra{\psi}A(t_s)^\dagger B(t_r)\ket{\psi}$, while $\rho_{01}=\rho_{10}^\dagger$. Hence, the desired cross-correlation function resides entirely in the coherence between the two selector branches.

The readout combines Hadamard-like measurement of the selector with Fourier measurement of the clock. We measure selector in the $X$ or $Y$ basis, with a respective $\pm1$ outcome denoted as $x$ or $y$, while the padded clock yields a frequency outcome $\omega$. Applying the same inversion in \cref{eq:lag} to the off-diagonal blocks gives
\begin{align}
G_{AB}(t_j)
=
\frac{N_t}{N_t-|j|}
\left(
\mathbb{E}\left[x e^{-it_j\omega}\right]
+i \mathbb{E}\left[y e^{-it_j\omega}\right]
\right),
\label{eq:lagAB}
\end{align}
for all $|j|\le N_t-1$, where 
the two expectation values are taken over their respective joint selector-clock outcomes.
Reading the selector in the $Z$ basis instead separates the diagonal blocks, recovering $G_{AA}$ and $G_{BB}$. We note that the sample complexity bound in  \cref{thm:samples} is asymptotically preserved with the selector outcome included. Only the constant prefactor is modified, leaving the overall scaling as $O(\log(N_t/\delta)/\epsilon^2 )$.

Furthermore, the two-operator construction extends to nonunitary probes through Pauli decomposition. After a Jordan-Wigner or Bravyi-Kitaev transformation, for example, fermionic annihilation and creation operators map to
\begin{align}
    c_i
    =
    \frac{1}{2}
    \left(
    \widetilde A_i-i\widetilde B_i
    \right), \qquad
    c_i^\dagger
    =
    \frac{1}{2}
    \left(
    \widetilde A_i+i\widetilde B_i
    \right),
\end{align}
where $\widetilde A_i$ and $\widetilde B_i$ denote Pauli strings. A single-particle Green's function is therefore a fixed linear combination of four Pauli string cross-correlations, each evaluated using \cref{eq:lagAB}.

\section{Comparison with Hadamard test}
\label{sec:compare_hadamard}

As a baseline for the resource comparison, we consider the Hadamard test that evaluates a correlation function one time lag at a time. As shown in \cref{fig:circ_comp_AB}, our algorithm instead elevates the insertion time to a quantum variable using the clock register, enabling many lags to be encoded in a single circuit while leaving Hamiltonian simulation uncontrolled. In this section, we quantify the resulting resource tradeoffs, first agnostic to the underlying time evolution routine and then for two representative choices: product formulas and qubitization. 

We target $G_{AA}(t_j)$ for lags $0<j \le N_t/2$ to additive precision $\epsilon$ with failure probability $\delta$ (negative lags follow from $G_{AA}(-t)=G_{AA}(t)^\ast$). The time record has length $(N_t - 1)\Delta t$. The Hadamard test measures the ancilla in both the $X$ and $Y$ bases to extract the real and imaginary parts of $G_{AA}(t_j)$. In contrast, our clock protocol recovers both from the same Fourier-basis measurement data. We account for the difference explicitly in our comparison.

\subsection{Cost models}
\label{subsec:cost_models}
Our comparison considers four resource metrics. $M$ is the total number of circuit executions, including all state preparations and measurements. $T$ is the maximal evolution time in a single execution, and $T_{\rm alg}$ is the total algorithmic runtime summed over all executions. Finally, $Q$ is the Hamiltonian simulation cost, counted in Trotter steps or block encoding queries. While $M$, $T$, and $T_{\rm alg}$ are independent of the simulation routine, $Q$ reflects how the time evolution is implemented. For $Q$, we revisit the product-formula and qubitization cost models introduced within \cref{subsec:resources}. In shorthand, an evolution of duration $t$ to accuracy $\varepsilon$ costs
\begin{align}
    Q_{\rm{PF}}
    =
    O\!\left(
    \frac{t^{1+1/d}}{\varepsilon^{1/d}}
    \right), \quad
    Q_{\rm{QSP}}
    = O\!\left(\alpha t + \mathcal{L}\left(\frac{1}{\varepsilon} \right)
    \right),
\end{align}
where $\alpha$ is the block encoding normalization and $\mathcal{L}(z)= \log z / \log\log z$.

The Hadamard-test and clock-based circuits enter the cost models differently. A Hadamard test evolves in one uninterrupted stretch to the target lag, whereas our clock construction divides the evolution into $N_t-1$ segments separated by operator insertions. For product formulas, subdividing the time interval does not change the leading scaling in the total interval length. However, qubitization acquires an additional $N_t$-dependent overhead since each segment incurs its own precision cost. This distinction is central in our subsequent analysis.

To compare the number of circuit executions, we define a \emph{sample round} as the executions needed to produce one sample for every target lag. The clock method completes a round in a single circuit execution. The Hadamard test runs at $N_t/2$ positive lags and in two ancilla measurement bases per lag, leading to $N_t$ executions per round. Both methods then require $O(\log(N_t/\delta)/\epsilon^2)$ such rounds for a simultaneous accuracy guarantee. The overall resource comparison therefore separates cleanly into the cost per round and the common number of round repetitions.

\subsection{Circuit executions and evolution time}
\label{subsec:executions}
We first examine the resources that do not depend on the Hamiltonian simulation routine. For Hadamard test, we estimate the real and imaginary parts of each target lag to accuracy $\epsilon/\sqrt{2}$. A union bound over the resulting $N_t$ estimates implies a round-repetition count of
\begin{align}\label{eq:htest_shotcount}
R_{\rm{HT}}
= \left\lceil
\frac{4}{\epsilon^2}
\log \frac{2N_t}{\delta}
\right\rceil.
\end{align}
Because each round contains $N_t$ circuits, $M_{\rm{HT}}=N_tR_{\rm{HT}}$. The longest run occurs at $j=N_t/2$ with
\begin{align}
T_{\rm{HT}}=\frac{N_t}{2}\Delta t,
\end{align}
and summing the evolution time over all lags, bases, and repetitions gives
\begin{align}
T_{\rm{alg}}^{\rm{HT}}
= 2R_{\rm{HT}}
\sum_{j=1}^{N_t/2} t_j \approx
\frac{N_t^2\Delta t}{\epsilon^2}
\ln\frac{2N_t}{\delta}.
\end{align}

For our clock construction, \cref{thm:samples} establishes
\begin{align}\label{eq:clk_shotcount}
R_{\rm{Clk}}
= \left\lceil
\frac{16}{\epsilon^2} \log \frac{4(N_t+1)}{\delta} \right\rceil.
\end{align}
As one circuit run samples the target lags, $M_{\rm{Clk}}=R_{\rm{Clk}}$. Each run sweeps through the entire time grid, so
\begin{align}
T_{\rm{Clk}}=(N_t-1)\Delta t
\end{align}
and
\begin{align}
T_{\rm{alg}}^{\rm{Clk}}
=
R_{\mathrm{Clk}}(N_t-1)\Delta t \approx
\frac{16N_t\Delta t}{\epsilon^2}
\ln\frac{4(N_t+1)}{\delta}.
\end{align}

Up to the slowly varying logarithmic factors, the clock construction reduces the circuit executions by a factor of $N_t/4$ and the total algorithmic runtime by $N_t/16$, at the cost of nearly doubling the maximum evolution time $T$. The $O(N_t)$ asymptotic gain is achieved by replacing measurements at each lag with a common sample stream. This overall advantage, on the other hand, is tempered by two constant-factor penalties. One comes from the Fej{\'e}r-window normalization, $N_t/(N_t-|j|)\le 2$, which appears quadratically in the sample complexity and contributes a factor of $4$ in $M_{\rm Clk}$ relative to $M_{\rm HT}$. An additional  factor of $4$ appears in $T_{\rm alg}$ because any clock run traverses the full record, whereas the average Hadamard-test run over $j=1,\ldots,N_t/2$ lasts approximately $N_t \Delta t/4$.

\subsection{Implementation of Hamiltonian simulation}
\label{subsec:simulation_compare}

We next compare the Hamiltonian simulation cost $Q$, where the uninterrupted evolution of the Hadamard test and the segmented evolution of our clock-based protocol give different resource scalings.

\emph{Trotterization.} For product formulas, our algorithm introduces no additional asymptotic cost due to splitting the time evolution into short segments. By \cref{eq:trotter}, one clock sample round has a cost $O((N_t\Delta t)^{1+1/d}/\varepsilon^{1/d})$. By contrast, one Hadamard-test sample round contains two circuits for each lag $t_j=j\Delta t$ with $j=1,\ldots,N_t/2$. Thus summing their costs gives 
\begin{align}
    2\sum_{j=1}^{N_t/2} \frac{(j\Delta t)^{1+1/d}}{\varepsilon^{1/d}}  = O\left( \frac{N_t (N_t\Delta t)^{1+1/d}}{(2+1/d) 2^{1+1/d}
\varepsilon^{1/d}} \right)
\end{align}
per round.

Multiplying by the repetition counts within \cref{subsec:executions}, we arrive at the Trotterization simulation costs:
\begin{align}
Q_{\rm PF}^{\rm Clk}
&=
\frac{(N_t\Delta t)^{1+1/d}}
{\varepsilon^{1/d}}
\frac{16 \log(4 N_t/\delta)}{\epsilon^2}, \\
Q_{\rm PF}^{\rm HT}
&=
\frac{
N_t (N_t\Delta t)^{1+1/d}
}{
(2+1/d)\,2^{1+1/d}\,
\varepsilon^{1/d}
} \frac{4 \log(2 N_t/\delta)}{\epsilon^2},
\end{align}
for our clock method and the Hadamard test respectively. Keeping the leading constants from the sampling bounds,
\begin{align}
\frac{Q_{\rm PF}^{\rm HT}}
     {Q_{\rm PF}^{\rm Clk}} \simeq \frac{N_t}{4(2+1/d)2^{1+1/d}},
\end{align}
up to constants in the logarithmic factors. At any fixed product-formula order, our clock algorithm hence attains a $\Theta(N_t)$ reduction in the Hamiltonian simulation cost. At fixed $\Delta t$, the total cost scales as $O(N_t^{1+1/d}\log N_t)$, or $O(N_t^{1/d}\log N_t)$ per lag, versus $O(N_t^{2+1/d}\log N_t)$ for the Hadamard test.

\emph{Qubitization.} Unlike with product formulas, splitting the evolution now places a precision-dependent overhead on each segment. 
We recall that the clock sweep consists of $N_t-1=O(N_t)$ segments, all implemented to accuracy $\varepsilon/N_t$. 
By \cref{eq:qubitization} one sample round of our clock protocol costs $O(\alpha N_t\Delta t + N_t \mathcal{L}(N_t/\varepsilon) )$. A Hadamard-test sample round makes $N_t$ simulation calls, but each performed to accuracy $\varepsilon$. Their durations sum to $O(N_t^2\Delta t)$, so the cost per round is $O (\alpha N_t^2\Delta t + N_t \mathcal{L}(1/\varepsilon))$.

\begin{figure*}[t]
    \centering
    \begin{minipage}[b]{0.33\textwidth}
        \centering
        \includegraphics[height=1.3\textwidth]{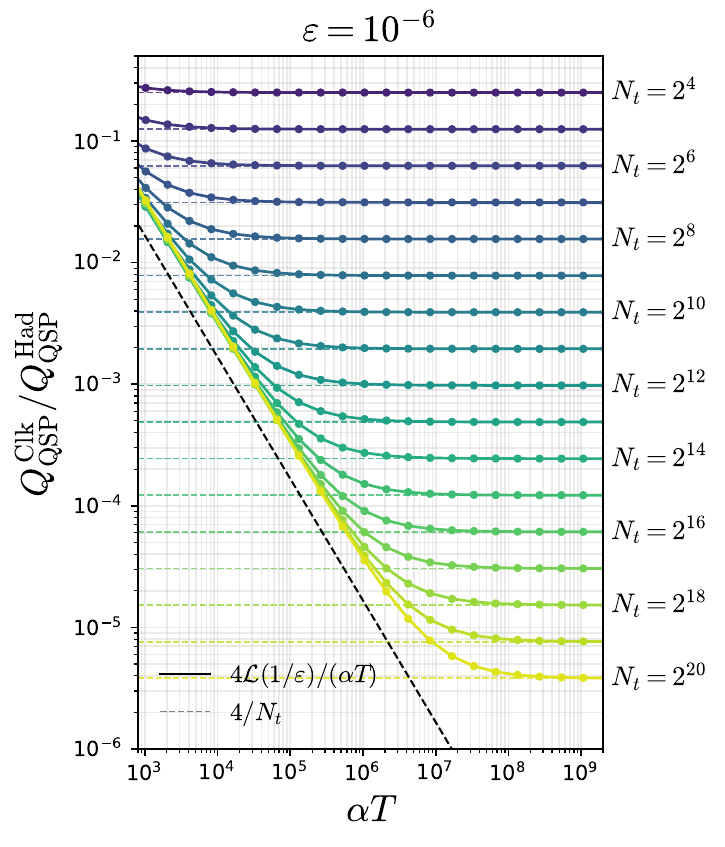}
        \\[-2pt]\hspace*{25pt}\small (a)
    \end{minipage}\hfill
    \begin{minipage}[b]{0.33\textwidth}
        \centering
        \hspace*{15pt}\includegraphics[height=1.3\textwidth]{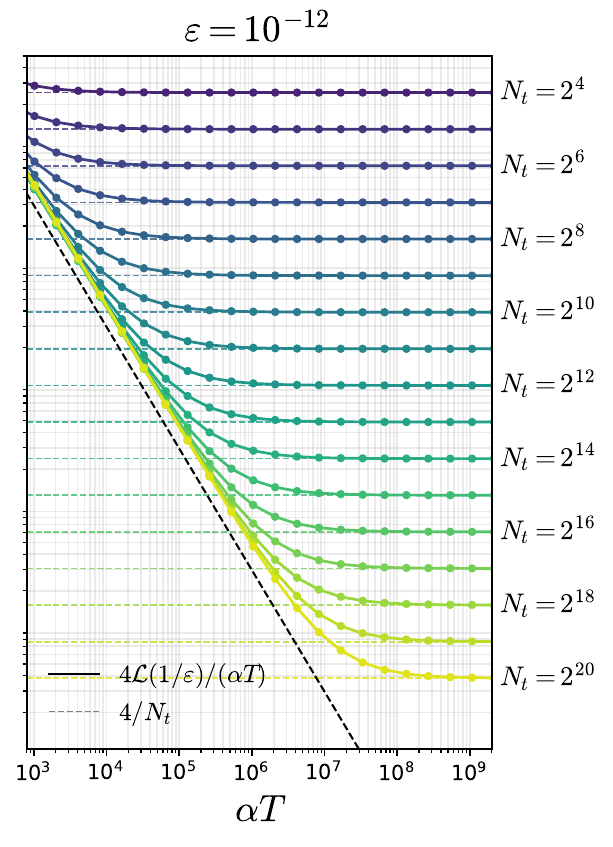}
        \\[-2pt]\hspace*{-5pt}\small (b)
    \end{minipage}\hfill
    \begin{minipage}[b]{0.33\textwidth}
        \centering
        \includegraphics[height=1.3\textwidth]{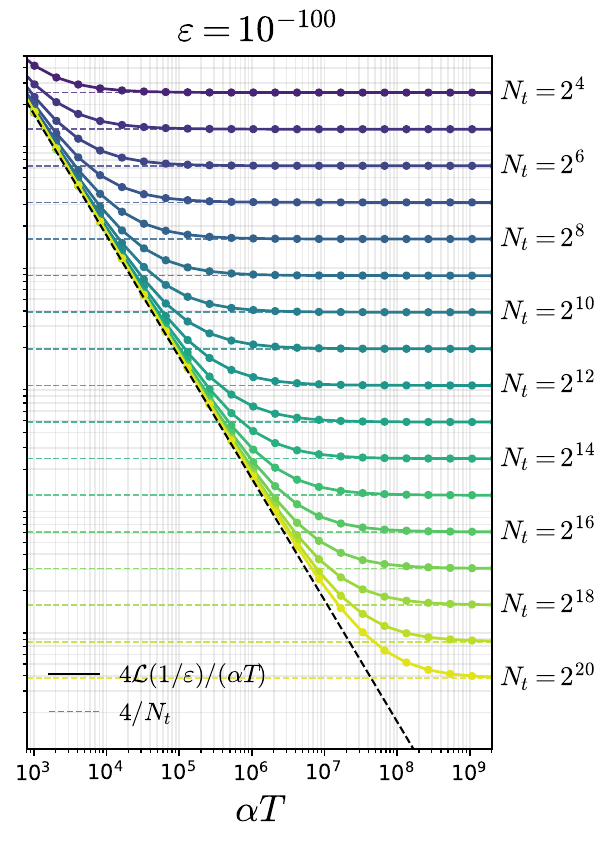}
        \\[-2pt]\hspace*{-25pt}\small (c)
    \end{minipage}
    \caption{Ratio of the qubitization-based Hamiltonian simulation costs of the clock and Hadamard-test algorithms, $Q_{\rm QSP}^{\rm Clk}/{Q_{\rm QSP}^{\rm HT}}$, as a function of $\alpha T$. The panels show simulation accuracies (a) $\varepsilon = 10^{-6}$, (b) $\varepsilon = 10^{-12}$, and (c) $\varepsilon = 10^{-100}$. Different curves correspond to different values of $N_t$. Each ratio initially follows the $O(1/(\alpha T))$ scaling in the $\mathcal{L}(\varepsilon) \ll \alpha T \ll N_t\mathcal{L}(N_t/\varepsilon)$ regime, and eventually converges to  the $4/N_t$ scaling in the $\alpha T \gg N_t\mathcal{L}(N_t/\varepsilon)$ regime. At $\varepsilon = 10^{-6}$ and $\varepsilon = 10^{-12}$, we observe that the clock algorithm reduces the quantum simulation resources by up to two orders of magnitude for $\alpha T \gtrsim 10^4$.}
    \label{fig:Qratios}
\end{figure*}

Accounting for the repetition counts from \cref{subsec:executions} yields
\begin{align}
    Q_{\rm QSP}^{\rm Clk} &= \left[ \alpha N_t\Delta t + N_t  \mathcal{L}(N_t/\varepsilon) \right] \frac{16 \log(N_t/\delta)}{\epsilon^2}, \\
    Q_{\rm QSP}^{\rm HT} &= \left[ \alpha N_t^2\Delta t + N_t \mathcal{L}(1/\varepsilon) \right] \frac{4 \log(N_t/\delta)}{\epsilon^2},
\end{align}
for the clock construction and the Hadamard test. Their ratio, modulo logarithmic constants, is 
\begin{align}
\begin{split}
    \frac{Q_{\rm QSP}^{\rm Clk}}
     {Q_{\rm QSP}^{\rm HT}} &\simeq \frac{4[ \alpha N_t\Delta t + N_t  \mathcal{L}(N_t/\varepsilon)]}{\alpha N_t^2\Delta t + N_t \mathcal{L}(1/\varepsilon)} \\
     &\simeq \frac{4}{N_t}\:\frac{ \alpha T + N_t  \mathcal{L}(N_t/\varepsilon)}{\alpha  T +  \mathcal{L}(1/\varepsilon)}
\end{split}
\end{align}
The relative advantage is shaped by the balance between the evolution-time and precision-overhead contributions $\alpha T,\:\mathcal{L}(\varepsilon)$, and $N_t\mathcal{L}(N_t/\varepsilon)$, which proposes three limits.

Before analyzing these limits, it is useful to establish the practical scales of $\alpha,T,N_t$, and $\varepsilon$.
The block encoding normalization $\alpha$, \emph{i.e.}, the sum of absolute coefficients of the Pauli terms in the Hamiltonian, is a system-extensive quantity that grows with both system size and interaction strength. For local many-body models, $\alpha = O(L)$ where $L$ is the system size, and for molecular Hamiltonians $\alpha = O(\mathrm{poly}(L))$. Large systems can therefore have $\alpha$ values in the hundreds or thousands. This growth in $\alpha$, however, does not extend to the frequencies relevant to correlation function under local Hamiltonians: energy scales induced by local perturbations are often intensive $O(1)$ quantities, up to finite-size corrections. Capturing these frequencies requires $\Delta t = O(1)$ by the Nyquist theorem, so we have $T \sim N_t \Delta t = O(N_t)$. The number of time points is set by the desired frequency resolution. $N_t \sim 10^2$-$10^3$, together with $\alpha \sim 10^2$-$10^3$, hence yields the representative range $\alpha T \simeq 10^4$-$10^6$. Finally, the time evolution precision $\varepsilon$ and measurement precision $\epsilon$ contribute a total error bounded by $\epsilon + 2\varepsilon$. There is therefore little benefit in adopting an extremely small $\varepsilon$ deep below $\epsilon$.
For material simulations, it is generally sufficient to achieve a total error of $10^{-6}$. We will consider $\varepsilon \simeq 10^{-6}$ as a practical value, while also examining smaller $\varepsilon$ values to explore the more extreme asymptotic limit.

\cref{fig:Qratios} plots the estimated quantum resource ratios for $\varepsilon = 10^{-6},10^{-12}$,$10^{-100}$ within panels (a)-(c) respectively, across $N_t = 2^4, 2^5, \ldots, 2^{20}$. For each $N_t$ and $\varepsilon$, the ratio of clock-algorithm resources to Hadamard-test resources,  $Q_{\rm QSP}^{\rm Clk}/{Q_{\rm QSP}^{\rm HT}}$, exhibits an initial $O(1/(\alpha T))$ scaling and then approaches an $ O(\Delta t/T)$ asymptote at large $\alpha T$. To understand this behavior, we distinguish the three limits defined by the interplay of evolution-time and precision-overhead terms $\alpha T$, $\mathcal{L}(\varepsilon)$, and $N_t\mathcal{L}(N_t/\varepsilon)$.

\textit{Case 1:} $\alpha T \gg N_t\mathcal{L}(N_t/\varepsilon)$. Here, the cost is dominated by evolution time. The clock scales as $O(N_t)$ per round while the Hadamard test scales as $O(N_t^2)$, yielding
\begin{align}
\begin{split}
    \frac{Q_{\rm QSP}^{\rm Clk}}
     {Q_{\rm QSP}^{\rm HT}}  &\simeq \frac{4 \alpha T}{N_t \alpha T} \simeq \frac{4}{N_t}, 
\end{split}
\end{align}
\emph{i.e.}, a $\Theta(N_t)$ reduction in Hamiltonian simulation cost. This is the asymptotic value approached by each curve in \cref{fig:Qratios} as $\alpha T \to \infty$. As seen in panels (a)-(c), the onset of convergence occurs at progressively larger $\alpha T$ for smaller $\varepsilon$.

\textit{Case 2:} $\mathcal{L}(\varepsilon) \ll \alpha T \ll N_t\mathcal{L}(N_t/\varepsilon)$. In this crossover, the Hadamard test is dominated by simulation time, and the clock algorithm is dominated by precision overhead. Thus,
\begin{align}
\begin{split}
    \frac{Q_{\rm QSP}^{\rm Clk}}
     {Q_{\rm QSP}^{\rm HT}}  &\simeq \frac{4 N_t \mathcal{L}(N_t/\varepsilon)}{N_t \alpha T} \simeq \frac{4 \mathcal{L}(N_t/\varepsilon)}{\alpha T}. 
\end{split}
\end{align}
This $O(1/(\alpha T))$ behavior is evident on the left-hand side of the curves in \cref{fig:Qratios}. Curves with different $N_t$ nearly overlap on the log-log scale because the prefactor $\mathcal{L}(N_t/\varepsilon)$ varies only logarithmically with $N_t$.

In this limit, when $\Delta t$ is fixed, the ratio of costs scales as $O(\log N_t /N_t)$ as $N_t \to \infty$. The clock algorithm therefore holds, up to logarithmic factor, a linear advantage in $N_t$. However, when $T$ is fixed instead, $\mathcal{L}(N_t/\varepsilon)$ continues to grow with $N_t$ and the Hadamard test will ultimately be less expensive. The critical transition happens at roughly $N_t^\ast = e^{\alpha T / 4} \varepsilon$, where its precise location depends on $\varepsilon$ and $\alpha T$. As discussed above and illustrated in \cref{fig:Qratios}, relevant values of $\alpha T$ and $\varepsilon$ imply $N_t^* \simeq 10^{11000} \varepsilon$. This places the transition far beyond the practical range of $N_t \le 2^{20}$ and $\varepsilon \ge 10^{-100}$.

\textit{Case 3:} $\alpha T \ll \mathcal{L}(1/\varepsilon)$. The cost is then dominated by precision overhead for both algorithms. A clock protocol now loses its linear advantage in $N_t$, and the relative cost in this limit is governed by $\mathcal{L}(N_t/\varepsilon)/
\mathcal{L}(1/\varepsilon)$. Particularly, for $\varepsilon=N_t^{-p}$ with $p>0$, 
\begin{align}
    \frac{Q_{\rm QSP}^{\rm Clk}}
     {Q_{\rm QSP}^{\rm HT}} \simeq \frac{\mathcal{L}(N_t/\varepsilon)}{\mathcal{L}(1/\varepsilon)} \rightarrow 1+\frac{1}{p},
\end{align}
so a constant factor separation remains asymptotically. For fixed $\varepsilon$, the cost ratio scales as $\Theta(\log N_t/\log \log N_t)$, thus benefiting the Hadamard test. Reaching this regime, nonetheless, requires extremely high simulation precision for the parameter range typical of materials simulations. For $\alpha T > 10^4$, $\alpha T \ll \mathcal{L}(1/\varepsilon)$ comes into play only around $\varepsilon^* = 10^{-250,000}$, which is essentially unattainable.

In summary, Cases 1 and 2 cover the only limits of practical relevance, and both decisively favor our clock algorithm over the Hadamard test by a factor of $\Theta(N_t)$ and $\Theta(N_t/\log N_t)$ in resource requirements.
We further evaluate this advantage numerically in \cref{sec:examples}.

\section{Numerical Example}
\label{sec:examples}
In this section, we numerically test the performance of our clock-based algorithm using an interacting fermionic system. We focus on its behavior at finite shot counts and compare its resource requirements with those discussed in \cref{sec:algorithm,sec:compare_hadamard}.

\subsection{Spinless Hubbard chain}
\label{subsec:hubbard_setup}
We consider the spinless Fermi-Hubbard model on a one-dimensional chain of $L$ sites,
\begin{align}
    H  = - J \sum_{i=1}^{L-1}
    \left(c_i^\dagger c_{i+1}+c_{i+1}^\dagger c_i\right)
    +V\sum_{i=1}^{L-1}n_i n_{i+1},
    \label{eq:tV}
\end{align}
where $c_i$ ($c_i^\dagger$) is the annihilation (creation) operator at site $i$, and $n_i=c_i^\dagger c_i$ the corresponding number operator. The two Hubbard parameters $J$ and $V$ set, respectively, the hopping amplitude and nearest-neighbor interaction strength. Under the Jordan-Wigner transformation, $c_i = Z_1\cdots Z_{i-1}(X_i-iY_i)/2$.

Let $\ket{\psi}$ be the ground state at a fixed particle number. As the probe, we choose the Majorana operator,
\begin{align}
\gamma_i = c_i + c_i^\dagger = Z_1 \cdots Z_{i-1} X_i,
\label{eq:majorana}
\end{align}
which is Hermitian and squares to one, hence unitary, so it can be inserted directly and no selector qubit is needed. Its autocorrelation function
\begin{align}
    G_{\gamma\gamma}(t)
    =
    \bra{\psi}c_i(t)c_i^\dagger\ket{\psi} + \bra{\psi}c_i^\dagger(t)c_i\ket{\psi},
\end{align}
where the remaining terms vanish because $\ket{\psi}$ has fixed particle number. Thus $G_{\gamma\gamma}(t)$ combines the local particle and hole propagators. Its spectral function captures the single-particle density of states for adding or removing a fermion at site $i$, with the removal excitations reflected to positive frequency. As a further consequence, the particle and hole contributions can be separated to construct the retarded Green's function; this involves cross-correlation with the second Majorana operator, $\tilde{\gamma}_i = i(c_i - c_i^\dagger) = Z_1 \cdots Z_{i-1} Y_i$, which can be computed using our selector construction from \cref{subsec:two_ops}.

For our numerical results below, we measure $G_{\gamma\gamma}(t)$ on the middle site at half-filling with $J = V = 1$. We choose $L = 8$ sites for convenience, although the protocol and calculations scale directly to much larger system sizes. 

\subsection{Resource-matched estimation errors}

\begin{figure*}[t]
    \centering
    \includegraphics[width=0.95\linewidth]{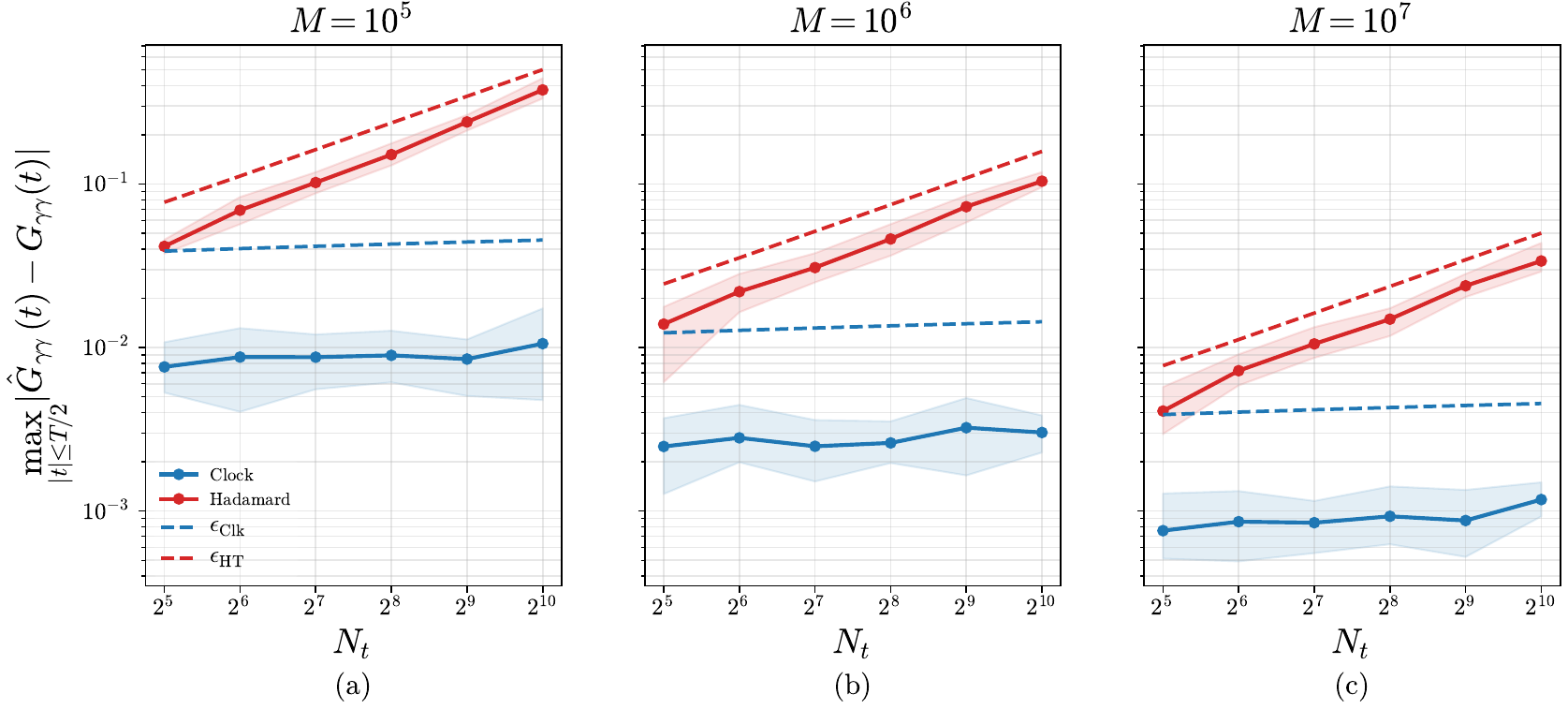}
    \caption{Theoretical and numerical errors of $G_{\gamma\gamma}(t)$ calculated with the clock and Hadamard-test protocols. The two approaches are compared using an equal qubitization query budget, with $M$ total shots for the clock algorithm and $2M$ total shots for the Hadamard test. We choose $M = 10^5$ in panel (a), $M = 10^6$ in panel (b), and $M = 10^7$ in panel (c). Both the Hadamard test and the clock algorithm were evaluated over 10 independent runs. The solid lines indicate the mean error, and the top (bottom) boundaries of the shaded regions indicate minimum (maximum) error. The dashed lines mark the theoretical precision given in \cref{eq:theoretical_precision}. The Hadamard-test error increases much more rapidly with $N_t$, while the clock-algorithm error remains relatively stable, pointing to a clear advantage at large $N_t$.  
    }
    \label{fig:Nt_dependence}
\end{figure*}

For comparison with the Hadamard test, we chose the qubitization-based evolution and match the Hamiltonian simulation resources of the two protocols.
Starting from a shot count $M$ for the clock algorithm, we determine for the Hadamard test the shot budget that gives an identical overall qubitization cost. To do so, we benchmark in the evolution-time-dominated regime (c.f. \cref{subsec:simulation_compare}) for two reasons. First, this regime suits the thermodynamic limit in which many-body phenomena (\emph{e.g.}, phase transitions) sharply emerge, while our results extend simply to large-scale fault-tolerant quantum simulation. Second, our $L=8$ example already provides a reasonable approximation, since we have $\alpha = (L-1)(2J+V) = 21$. With simulation accuracy $\varepsilon = 10^{-6}$ and duration $T = 1000$, this leads to $\alpha T = 2.1\times10^4$. As shown in \cref{fig:Qratios}(a), the resource ratios are quite close to their large-$\alpha T$ limits for $N_t\le2^{10}$. The same regime thereby serves as a natural benchmark even for our finite system.

At fixed total quantum resources, $M$ shots for the clock algorithm then correspond to $2M/N_t$ shots per lag for the Hadamard test. We consider $N_t=2^5$, $2^6$, $\dots$, $2^{10}$ time points spanning a maximal evolution time $T=N_t\Delta t=1000$, and we allocate $M = 10^5, 10^6$, and $10^7$ shots for the clock algorithm. The equivalent full shot budgets for the Hadamard test are $2M = 2 \times 10^5, 2 \times 10^6$, and $2 \times 10^7$. For each case, these shots are distributed over $N_t/2$ lags and $2$ ancilla measurement bases per lag. With these matched shot counts, the expected accuracies of our algorithm and the Hadamard test can be found from \cref{eq:clk_shotcount,eq:htest_shotcount} as the following:
\begin{align}\label{eq:theoretical_precision}
\begin{split}
    \epsilon_\mathrm{Clk} &= 4 \sqrt{\frac{\log(4N_t/\delta)}{M}},
    \\
    \epsilon_\mathrm{HT} &= 2 \sqrt{\frac{N_t \log(2N_t/\delta)}{2M}}.
\end{split}
\end{align}

\cref{fig:Nt_dependence} shows values of the maximum estimation error, $\max_{|t|\leq T/2}|\hat G_{\gamma\gamma}(t)-G_{\gamma\gamma}(t)|$, for the clock algorithm and the Hadamard test, together with their theoretical upper bounds. We notice that the numerical errors consistently track the predicted $N_t$-dependence in the upper bounds $\epsilon_\mathrm{Clk}$ and $\epsilon_\mathrm{HT}$, although these bounds appear conservative. Across simulation, the Hadamard-test error is on average $\sim 0.6 \epsilon_\mathrm{HT}$, while the clock error is $\sim 0.2 \epsilon_\mathrm{Clk}$. The actual error of our protocol is well below its theoretical ceiling.

Increasing $M$ improves the accuracy of both protocols, but their dependence on $N_t$ is qualitatively distinct. The clock error stays almost constant over the plotted range, reflecting its weak $\sqrt{\log(N_t)}$ dependence. The Hadamard test, however, has power-law error due to its $\sqrt{N_t\log(N_t)}$ dependence. Recall that this difference arises because the Hadamard-test budget must be split evenly among $O(N_t)$ individual circuits, whereas every clock shot contributes information about all time lags. As a consequence, the accuracy advantage of the clock algorithm grows with $N_t$.

\section{Discussion and Outlook}
\label{sec:discussion}

In this work, we introduced a clock-register approach to equilibrium correlation functions in which we promote the probe insertion time from a fixed circuit parameter to a quantum variable. The system undergoes the same uncontrolled Hamiltonian evolution on every clock branch, with the probe inserted at a time selected by that branch.
For a stationary state, the clock captures the correlation function at all time lags, where a Fourier measurement of the clock yields the corresponding finite-time-broadened spectral distribution. These relations are exact for both pure and mixed stationary states. Moreover, each Fourier outcome contributes simultaneously to every target lag, allowing all $N_t$ correlators to be estimated to a precision $\epsilon$ with $M = O(\log(N_t/\delta)/\epsilon^2)$ total shots. The lower bound of \cref{thm:lower} shows this scaling is optimal, up to constants, in the sample-access model. With a single selector qubit, the same construction applies to correlations between two different operator probes.

Compared with the Hadamard test, our construction replaces separate measurements at individual lags by a shared measurement record. This reduces the number of circuit executions by $O(N_t)$ and, in turn, yields an $O(N_t)$ saving in total runtime, up to constant and logarithmic factors. The improvement carries over directly to Trotter simulation. For qubitization, the relative cost depends on the balance between the physical evolution and precision overhead incurred from the $N_t$ evolution segments. For all practical cases of Hamiltonian simulation, we find that the evolution cost is dominated by simulation time, and in this regime the clock approach requires at best $O(N_t)$ (Case 1 in \cref{subsec:simulation_compare}), and at worst $O(N_t/\log N_t)$ (Case 2 in \cref{subsec:simulation_compare}) fewer quantum resources. 
Our numerical example demonstrates this finite-shot behavior in a spinless Hubbard model and directly compares the observed accuracy with theoretical prediction.

Several aspects of the construction appear favorable for practical implementation. The Hamiltonian evolution is uncontrolled, the QFT can be performed semiclassically \cite{Griffiths1996}, and unary iteration achieves probe insertions with a linear Toffoli cost~\cite{Babbush2018}. The measurement outcomes also reveal spectral information, such as linear functionals of the correlation function, without additional circuit runs. The protocol is therefore well suited to settings in which forward evolution is readily available but coherent control or reversal of the full simulation is demanding, including hybrid analog-digital platforms.

A further advantage of our algorithm, which was not discussed above, is its efficiency in repeated initial state preparation. While we assumed time evolution accounts for most of the quantum resources, each circuit execution requires a fresh copy of the stationary system state $\ket{\psi}$, whose preparation can itself be nontrivial. In particular, preparing ground states of generic Hamiltonians is known to be QMA-complete \cite{QMA}. Moreover, state-preparation strategies based on variational quantum eigensolver \cite{grimsley2019adaptive, tilly2022variational}, quantum phase estimation combined with amplitude amplification \cite{Lin2020}, and variations of quantum imaginary time evolution \cite{Motta2020, gluza2026double} can lead to deep circuits. Because the clock protocol reduces the number of total shots from $O(N_t\log(N_t/\delta)/\epsilon^2)$ for the standard Hadamard test to $O(\log(N_t/\delta)/\epsilon^2)$, it likewise reduces the number of state preparations by a factor of $O(N_t)$. This saving can hence be especially beneficial when preparing $\ket{\psi}$ is comparable to, or even more expensive than, the subsequent real-time evolution.

Beyond the applications considered here, the clock framework opens a broader route to quantum dynamical measurements. If a family of observables can be accessed by coherent circuit control of when or where an operator probe is inserted, many classically distinct measurements may be combined into a shared quantum experiment. Extensions to nonequilibrium Green’s functions and out-of-time-order (OTOC) correlators will require additional structure, because these settings do not generally inherit the Toeplitz form that is guaranteed by time translational invariance at equilibrium. Developing such extensions is an interesting direction for future studies. Moreover, our results show that promoting measurement parameters to quantum variables can reduce the number of circuit runs while leaving many-body evolution uncontrolled, offering a powerful paradigm for obtaining dynamical information from large-scale quantum simulations.

\begin{acknowledgments}
We acknowledge support from the U.S. Department of Energy (DOE) via Office of Science, under Contract No.~DE-AC02-05CH11231. YS and RVB were funded through Accelerated Research in Quantum Computing, Fundamental Algorithmic Research toward Quantum Utility (FAR-Qu). EK and WAdJ were funded through the Office of Advanced Scientific Computing Research’s Accelerated Research for Quantum Computing Program (FAR-QC).

\end{acknowledgments}

\bibliography{refs}

@article{francis2020quantum,
  title = {Quantum computation of magnon spectra},
  author = {Francis, Akhil and Freericks, J. K. and Kemper, A. F.},
  journal = {Phys. Rev. B},
  volume = {101},
  issue = {1},
  pages = {014411},
  numpages = {7},
  year = {2020},
  month = {Jan},
  publisher = {American Physical Society},
  doi = {10.1103/PhysRevB.101.014411},
  url = {https://link.aps.org/doi/10.1103/PhysRevB.101.014411}
}

@article{Low2017,
  author  = {Low, Guang Hao and Chuang, Isaac L.},
  title   = {Optimal {Hamiltonian} simulation by quantum signal processing},
  journal = {Phys. Rev. Lett.},
  volume  = {118},
  pages   = {010501},
  year    = {2017},
  doi     = {10.1103/PhysRevLett.118.010501}
}

@Article{Low2019,
  author    = {Low, Guang Hao and Chuang, Isaac L.},
  journal   = {{Quantum}},
  title     = {Hamiltonian {S}imulation by {Q}ubitization},
  year      = {2019},
  month     = jul,
  pages     = {163},
  volume    = {3},
  doi       = {10.22331/q-2019-07-12-163},
  publisher = {{Verein zur F{\"{o}}rderung des Open Access Publizierens in den Quantenwissenschaften}},
}

@Article{Lin2020,
  author    = {Lin, Lin and Tong, Yu},
  journal   = {{Quantum}},
  title     = {Near-optimal ground state preparation},
  year      = {2020},
  issn      = {2521-327X},
  month     = dec,
  pages     = {372},
  volume    = {4},
  doi       = {10.22331/q-2020-12-14-372},
  publisher = {{Verein zur F{\"{o}}rderung des Open Access Publizierens in den Quantenwissenschaften}},
}

@article{QMA,
author = {Kempe, Julia and Kitaev, Alexei and Regev, Oded},
doi = {10.1137/S0097539704445226},
journal = {SIAM Journal on Computing},
number = {5},
pages = {1070-1097},
title = {The Complexity of the Local Hamiltonian Problem},
volume = {35},
year = {2006}
}

@article{Roggero20,
  title = {Spectral-density estimation with the Gaussian integral transform},
  author = {Roggero, A.},
  journal = {Phys. Rev. A},
  volume = {102},
  issue = {2},
  pages = {022409},
  numpages = {12},
  year = {2020},
  month = {Aug},
  publisher = {American Physical Society},
  doi = {10.1103/PhysRevA.102.022409},
  url = {https://link.aps.org/doi/10.1103/PhysRevA.102.022409}
}

@article{Bauer16,
 author = {Bauer, Bela and Wecker, Dave and Millis, Andrew J. and Hastings, Matthew B. and Troyer, Matthias},
 doi = {10.1103/PhysRevX.6.031045},
 journal = {Phys. Rev. X},
 volume = {6},
 issue = {3},
 pages = {031045},
 numpages = {11},
 year = {2016},
 title = {{Hybrid quantum-classical approach to correlated materials}}
}

@article{Kosugi20linear,
  title = {Linear-response functions of molecules on a quantum computer: Charge and spin responses and optical absorption},
  author = {Kosugi, Taichi and Matsushita, Yu-Ichiro},
  journal = {Phys. Rev. Research},
  volume = {2},
  issue = {3},
  pages = {033043},
  numpages = {16},
  year = {2020},
  month = {Jul},
  publisher = {American Physical Society},
  doi = {10.1103/PhysRevResearch.2.033043},
  url = {https://link.aps.org/doi/10.1103/PhysRevResearch.2.033043}
}

@Article{Pedernales2014,
  author    = {Pedernales, J. S. and Di Candia, R. and Egusquiza, I. L. and Casanova, J. and Solano, E.},
  journal   = {Phys. Rev. Lett.},
  title     = {Efficient Quantum Algorithm for Computing $n$-time Correlation Functions},
  year      = {2014},
  month     = jul,
  pages     = {020505},
  volume    = {113},
  doi       = {10.1103/PhysRevLett.113.020505},
  issue     = {2},
  numpages  = {5},
  publisher = {American Physical Society},
}

@article{Tong21,
  title = {Fast inversion, preconditioned quantum linear system solvers, fast Green's-function computation, and fast evaluation of matrix functions},
  author = {Tong, Yu and An, Dong and Wiebe, Nathan and Lin, Lin},
  journal = {Phys. Rev. A},
  volume = {104},
  issue = {3},
  pages = {032422},
  numpages = {33},
  year = {2021},
  month = {Sep},
  publisher = {American Physical Society},
  doi = {10.1103/PhysRevA.104.032422},
  url = {https://link.aps.org/doi/10.1103/PhysRevA.104.032422}
}

@book{Mahan,
  author    = {Mahan, Gerald D.},
  title     = {Many-Particle Physics},
  publisher = {Kluwer Academic/Plenum Publishers},
  address   = {New York},
  year      = {2000},
  edition   = {3rd},
  doi       = {10.1007/978-1-4757-5714-9},
  url       = {https://doi.org/10.1007/978-1-4757-5714-9}
}

@article{grimsley2019adaptive,
  title={An adaptive variational algorithm for exact molecular simulations on a quantum computer},
  author={Grimsley, Harper R and Economou, Sophia E and Barnes, Edwin and Mayhall, Nicholas J},
  journal={Nat. Commun.},
  volume={10},
  number={1},
  pages={3007},
  year={2019},
  publisher={Nature Publishing Group UK London},
  doi={10.1038/s41467-019-10988-2}
}

@article{hogan2026efficient,
  title={Efficient quantum implementation of dynamical mean field theory for correlated materials},
  author={Hogan, Norman and K{\"o}kc{\"u}, Efekan and Steckmann, Thomas and Doak, Liam P and Mejuto-Zaera, Carlos and Camps, Daan and Van Beeumen, Roel and de Jong, Wibe A and Kemper, AF},
  journal={npj Computational Materials},
  year={2026},
  publisher={Nature Publishing Group UK London},
  url={https://doi.org/10.1038/s41524-026-02289-2}
}

@article{tilly2022variational,
  title={The variational quantum eigensolver: a review of methods and best practices},
  author={Tilly, Jules and Chen, Hongxiang and Cao, Shuxiang and Picozzi, Dario and Setia, Kanav and Li, Ying and Grant, Edward and Wossnig, Leonard and Rungger, Ivan and Booth, George H and others},
  journal={Phys. Rep.},
  volume={986},
  pages={1--128},
  year={2022},
  publisher={Elsevier},
  doi={10.1016/j.physrep.2022.08.003}
}

@article{endo2020calculation,
  author = {Endo, Suguru and Kurata, Iori and Nakagawa, Yuya O.},
  doi = {10.1103/PhysRevResearch.2.033281},
  issue = {3},
  journal = {Phys. Rev. Research},
  month = {Aug},
  pages = {033281},
  publisher = {American Physical Society},
  title = {Calculation of the {Green's} function on near-term quantum computers},
  url = {https://link.aps.org/doi/10.1103/PhysRevResearch.2.033281},
  volume = {2},
  year = {2020},
}

@article{gluza2026double,
  title={Double-bracket quantum algorithms for quantum imaginary-time evolution},
  author={Gluza, Marek and Son, Jeongrak and Tiang, Bi Hong and Zander, Ren{\'e} and Seidel, Raphael and Suzuki, Yudai and Holmes, Zo{\"e} and Ng, Nelly HY},
  journal={Phys. Rev. Lett.},
  volume={136},
  number={2},
  pages={020601},
  year={2026},
  publisher={APS},
  doi={10.1103/rw81-k8vk}
}

@article{vanHove1954,
  author  = {Van Hove, L{\'e}on},
  title   = {Correlations in space and time and {Born} approximation scattering in systems of interacting particles},
  journal = {Phys. Rev.},
  volume  = {95},
  pages   = {249},
  year    = {1954},
  doi     = {10.1103/PhysRev.95.249}
}

@book{Squires2012,
  author    = {Squires, G. L.},
  title     = {{Introduction to the Theory of Thermal Neutron Scattering}},
  publisher = {Cambridge University Press},
  address   = {Cambridge},
  year      = {2012},
  edition   = {3rd},
  doi       = {10.1017/CBO9781139107808},
  url       = {https://doi.org/10.1017/CBO9781139107808}
}

@article{Kubo1957,
  author  = {Kubo, Ryogo},
  title   = {Statistical-mechanical theory of irreversible processes. {I}. General theory and simple applications to magnetic and conduction problems},
  journal = {J. Phys. Soc. Jpn.},
  volume  = {12},
  pages   = {570},
  year    = {1957},
  doi     = {10.1143/JPSJ.12.570}
}

@article{Damascelli2003,
  author  = {Damascelli, Andrea and Hussain, Zahid and Shen, Zhi-Xun},
  title   = {Angle-resolved photoemission studies of the cuprate superconductors},
  journal = {Rev. Mod. Phys.},
  volume  = {75},
  pages   = {473},
  year    = {2003},
  doi     = {10.1103/RevModPhys.75.473}
}

@article{Callen1951,
  author  = {Callen, Herbert B. and Welton, Theodore A.},
  title   = {Irreversibility and generalized noise},
  journal = {Phys. Rev.},
  volume  = {83},
  pages   = {34},
  year    = {1951},
  doi     = {10.1103/PhysRev.83.34}
}

@book{Bruus2004,
  author    = {Bruus, Henrik and Flensberg, Karsten},
  title     = {{Many-Body Quantum Theory in Condensed Matter Physics: An Introduction}},
  publisher = {Oxford University Press},
  address   = {Oxford},
  year      = {2004},
  doi       = {10.1093/oso/9780198566335.001.0001},
  url       = {https://doi.org/10.1093/oso/9780198566335.001.0001}
}

@book{Stefanucci2013,
  author    = {Stefanucci, Gianluca and van Leeuwen, Robert},
  title     = {{Nonequilibrium Many-Body Theory of Quantum Systems: A Modern Introduction}},
  publisher = {Cambridge University Press},
  address   = {Cambridge},
  year      = {2013},
  doi       = {10.1017/CBO9781139023979},
  url       = {https://doi.org/10.1017/CBO9781139023979}
}

@article{Georges1996,
  author  = {Georges, Antoine and Kotliar, Gabriel and Krauth, Werner and Rozenberg, Marcelo J.},
  title   = {Dynamical mean-field theory of strongly correlated fermion systems and the limit of infinite dimensions},
  journal = {Rev. Mod. Phys.},
  volume  = {68},
  pages   = {13},
  year    = {1996},
  doi     = {10.1103/RevModPhys.68.13}
}

@article{Kotliar2006,
  author  = {Kotliar, Gabriel and Savrasov, Sergej Y. and Haule, Kristjan and Oudovenko, Vladimir S. and Parcollet, Olivier and Marianetti, Chris A.},
  title   = {Electronic structure calculations with dynamical mean-field theory},
  journal = {Rev. Mod. Phys.},
  volume  = {78},
  pages   = {865},
  year    = {2006},
  doi     = {10.1103/RevModPhys.78.865}
}

@article{Rungger2019,
author        = {Rungger, I. and Fitzpatrick, N. and Chen, H. and Alderete, C. H.
and Apel, H. and Cowtan, A. and Patterson, A. and Munoz Ramo, D.
and Zhu, Y. and Nguyen, N. H. and Grant, E. and Chretien, S.
and Wossnig, L. and Linke, N. M. and Duncan, R.},
title         = {Dynamical mean field theory algorithm and experiment on quantum computers},
journal       = {arXiv},
year          = {2019},
doi           = {10.48550/arXiv.1910.04735},
url           = {https://arxiv.org/abs/1910.04735}
}

@article{Feynman1982,
  author  = {Feynman, Richard P.},
  title   = {Simulating physics with computers},
  journal = {Int. J. Theor. Phys.},
  volume  = {21},
  pages   = {467},
  year    = {1982},
  doi     = {10.1007/BF02650179}
}

@article{Lloyd1996,
  author  = {Lloyd, Seth},
  title   = {Universal quantum simulators},
  journal = {Science},
  volume  = {273},
  pages   = {1073},
  year    = {1996},
  doi     = {10.1126/science.273.5278.1073}
}

@article{Childs2021,
  author  = {Childs, Andrew M. and Su, Yuan and Tran, Minh C. and Wiebe, Nathan and Zhu, Shuchen},
  title   = {Theory of {Trotter} error with commutator scaling},
  journal = {Phys. Rev. X},
  volume  = {11},
  pages   = {011020},
  year    = {2021},
  doi     = {10.1103/PhysRevX.11.011020}
}

@article{Motta2020,
  author  = {Motta, Mario and Sun, Chong and Tan, Adrian T. K. and O'Rourke, Matthew J. and Ye, Erika and Minnich, Austin J. and Brand{\~a}o, Fernando G. S. L. and Chan, Garnet Kin-Lic},
  title   = {Determining eigenstates and thermal states on a quantum computer using quantum imaginary time evolution},
  journal = {Nat. Phys.},
  volume  = {16},
  pages   = {205},
  year    = {2020},
  doi     = {10.1038/s41567-019-0704-4}
}

@book{Dalzell2023,
author    = {Dalzell, Alexander M. and McArdle, Sam and Berta, Mario and Bienias, Przemyslaw and Chen, Chi-Fang and Gily{\'e}n, Andr{\'a}s
and Hann, Connor T. and Kastoryano, Michael J.
and Khabiboulline, Emil T. and Kubica, Aleksander
and Salton, Grant and Wang, Samson
and Brand{\~a}o, Fernando G. S. L.},
title     = {Quantum Algorithms: A Survey of Applications and End-to-End Complexities},
publisher = {Cambridge University Press},
year      = {2025},
doi       = {10.1017/9781009639651},
url       = {https://doi.org/10.1017/9781009639651}
}

@article{Ortiz2001,
  author  = {Ortiz, G. and Gubernatis, J. E. and Knill, E. and Laflamme, R.},
  title   = {Quantum algorithms for fermionic simulations},
  journal = {Phys. Rev. A},
  volume  = {64},
  pages   = {022319},
  year    = {2001},
  doi     = {10.1103/PhysRevA.64.022319}
}

@article{Somma2002,
  author  = {Somma, R. and Ortiz, G. and Gubernatis, J. E. and Knill, E. and Laflamme, R.},
  title   = {Simulating physical phenomena by quantum networks},
  journal = {Phys. Rev. A},
  volume  = {65},
  pages   = {042323},
  year    = {2002},
  doi     = {10.1103/PhysRevA.65.042323}
}

@article{Kreula2016,
  author  = {Kreula, J. M. and Garc{\'i}a-{\'A}lvarez, L. and Lamata, L. and Clark, S. R. and Solano, E. and Jaksch, D.},
  title   = {Few-qubit quantum-classical simulation of strongly correlated lattice fermions},
  journal = {EPJ Quantum Technol.},
  volume  = {3},
  pages   = {11},
  year    = {2016},
  doi     = {10.1140/epjqt/s40507-016-0049-1}
}

@article{Keen2020,
  author  = {Keen, Trevor and Maier, Thomas and Johnston, Steven and Lougovski, Pavel},
  title   = {Quantum-classical simulation of two-site dynamical mean-field theory on noisy quantum hardware},
  journal = {Quantum Sci. Technol.},
  volume  = {5},
  pages   = {035001},
  year    = {2020},
  doi     = {10.1088/2058-9565/ab7d4c}
}

@article{Libbi2022,
  author  = {Libbi, Francesco and Rizzo, Jacopo and Tacchino, Francesco and Marzari, Nicola and Tavernelli, Ivano},
  title   = {Effective calculation of the {Green}'s function in the time domain on near-term quantum processors},
  journal = {Phys. Rev. Research},
  volume  = {4},
  pages   = {043038},
  year    = {2022},
  doi     = {10.1103/PhysRevResearch.4.043038}
}

@article{Steckmann2023,
  author  = {Steckmann, Thomas and Keen, Trevor and K{\"o}kc{\"u}, Efekan and Kemper, Alexander F. and Dumitrescu, Eugene F. and Wang, Yan},
  title   = {Mapping the metal-insulator phase diagram by algebraically fast-forwarding dynamics on a cloud quantum computer},
  journal = {Phys. Rev. Research},
  volume  = {5},
  pages   = {023198},
  year    = {2023},
  doi     = {10.1103/PhysRevResearch.5.023198}
}

@article{chiesa2019quantum,
  author  = {Chiesa, A. and Tacchino, F. and Grossi, M. and Santini, P. and Tavernelli, I. and Gerace, D. and Carretta, S.},
  title   = {Quantum hardware simulating four-dimensional inelastic neutron scattering},
  journal = {Nat. Phys.},
  volume  = {15},
  pages   = {455},
  year    = {2019},
  doi     = {10.1038/s41567-019-0437-4}
}

@article{Kokcu2024,
  author  = {K{\"o}kc{\"u}, Efekan and Labib, Heba A. and Freericks, J. K. and Kemper, A. F.},
  title   = {A linear response framework for quantum simulation of bosonic and fermionic correlation functions},
  journal = {Nat. Commun.},
  volume  = {15},
  pages   = {3881},
  year    = {2024},
  doi     = {10.1038/s41467-024-47729-z}
}

@article{Lu2021,
  author  = {Lu, Sirui and Ba{\~n}uls, Mari Carmen and Cirac, J. Ignacio},
  title   = {Algorithms for quantum simulation at finite energies},
  journal = {PRX Quantum},
  volume  = {2},
  pages   = {020321},
  year    = {2021},
  doi     = {10.1103/PRXQuantum.2.020321}
}

@article{Keen2021,
author        = {Keen, Trevor and Dumitrescu, Eugene and Wang, Yan},
title         = {Quantum algorithms for ground-state preparation and {Green}'s function calculation},
journal       = {arXiv},
year          = {2021},
archivePrefix = {arXiv},
doi           = {10.48550/arXiv.2112.05731},
url           = {https://arxiv.org/abs/2112.05731}
}

@article{Chen2021,
  author  = {Chen, Hongxiang and Nusspickel, Max and Tilly, Jules and Booth, George H.},
  title   = {Variational quantum eigensolver for dynamic correlation functions},
  journal = {Phys. Rev. A},
  volume  = {104},
  pages   = {032405},
  year    = {2021},
  doi     = {10.1103/PhysRevA.104.032405}
}

@article{Huang2022jpcl,
  author  = {Huang, Kaixuan and Cai, Xiaoxia and Li, Hao and Ge, Zi-Yong and Hou, Ruijuan and Li, Hekang and Liu, Tong and Shi, Yunhao and Chen, Chitong and Zheng, Dongning and Xu, Kai and Liu, Zhi-Bo and Li, Zhendong and Fan, Heng and Fang, Wei-Hai},
  title   = {Variational quantum computation of molecular linear response properties on a superconducting quantum processor},
  journal = {J. Phys. Chem. Lett.},
  volume  = {13},
  pages   = {9114},
  year    = {2022},
  doi     = {10.1021/acs.jpclett.2c02381}
}

@article{Jensen2023,
  author  = {Jensen, Phillip W. K. and Johnson, Peter D. and Kunitsa, Alexander A.},
  title   = {Near-term quantum algorithm for computing molecular and materials properties based on recursive variational series methods},
  journal = {Phys. Rev. A},
  volume  = {108},
  pages   = {022422},
  year    = {2023},
  doi     = {10.1103/PhysRevA.108.022422}
}

@article{Jamet2022,
  title         = {Quantum subspace expansion algorithm for {Green's} functions},
  author        = {Jamet, Francois and Agarwal, Abhishek and Rungger, Ivan},
  journal       = {arXiv},
  year          = {2022},
  archivePrefix = {arXiv},
  doi           = {10.48550/arXiv.2205.00094},
  url           = {https://arxiv.org/abs/2205.00094}
}

@article{Somma2019,
  author  = {Somma, Rolando D.},
  title   = {Quantum eigenvalue estimation via time series analysis},
  journal = {New J. Phys.},
  volume  = {21},
  pages   = {123025},
  year    = {2019},
  doi     = {10.1088/1367-2630/ab5c60}
}

@article{Lin2022,
  author  = {Lin, Lin and Tong, Yu},
  title   = {{Heisenberg}-limited ground-state energy estimation for early fault-tolerant quantum computers},
  journal = {PRX Quantum},
  volume  = {3},
  pages   = {010318},
  year    = {2022},
  doi     = {10.1103/PRXQuantum.3.010318}
}

@article{Dong2022,
  author  = {Dong, Yulong and Lin, Lin and Tong, Yu},
  title   = {Ground-state preparation and energy estimation on early fault-tolerant quantum computers via quantum eigenvalue transformation of unitary matrices},
  journal = {PRX Quantum},
  volume  = {3},
  pages   = {040305},
  year    = {2022},
  doi     = {10.1103/PRXQuantum.3.040305}
}

@article{Stroeks2022,
  author  = {Stroeks, M. E. and Helsen, J. and Terhal, B. M.},
  title   = {Spectral estimation for {Hamiltonians}: a comparison between classical imaginary-time evolution and quantum real-time evolution},
  journal = {New J. Phys.},
  volume  = {24},
  pages   = {103024},
  year    = {2022},
  doi     = {10.1088/1367-2630/ac919c}
}

@article{Hoeffding1963,
  author  = {Hoeffding, Wassily},
  title   = {Probability inequalities for sums of bounded random variables},
  journal = {J. Am. Stat. Assoc.},
  volume  = {58},
  pages   = {13},
  year    = {1963},
  doi     = {10.1080/01621459.1963.10500830}
}

@article{Gidney2018,
  author  = {Gidney, Craig},
  title   = {Halving the cost of quantum addition},
  journal = {Quantum},
  volume  = {2},
  pages   = {74},
  year    = {2018},
  doi     = {10.22331/q-2018-06-18-74}
}

@article{Nie2024,
author        = {Nie, Junhong and Zi, Wei and Sun, Xiaoming},
title         = {Quantum circuit for multi-qubit {Toffoli} gate with optimal resource},
journal       = {arXiv},
year          = {2024},
archivePrefix = {arXiv},
doi           = {10.48550/arXiv.2402.05053},
url           = {https://arxiv.org/abs/2402.05053}
}

@article{Khattar2025,
  doi = {10.22331/q-2025-05-21-1752},
  url = {https://doi.org/10.22331/q-2025-05-21-1752},
  title = {Rise of conditionally clean ancillae for efficient quantum circuit constructions},
  author = {Khattar, Tanuj and Gidney, Craig},
  journal = {{Quantum}},
  issn = {2521-327X},
  publisher = {{Verein zur F{\"{o}}rderung des Open Access Publizierens in den Quantenwissenschaften}},
  volume = {9},
  pages = {1752},
  month = may,
  year = {2025}
}

@article{Kitaev1995,
author        = {Kitaev, A. Yu.},
title         = {Quantum measurements and the {Abelian} stabilizer problem},
journal       = {arXiv},
year          = {1995},
archivePrefix = {arXiv},
doi           = {10.48550/arXiv.quant-ph/9511026},
url           = {https://arxiv.org/abs/quant-ph/9511026}
}

@article{Cleve1998,
  author  = {Cleve, Richard and Ekert, Artur and Macchiavello, Chiara and Mosca, Michele},
  title   = {Quantum algorithms revisited},
  journal = {Proc. R. Soc. Lond. A},
  volume  = {454},
  pages   = {339},
  year    = {1998},
  doi     = {10.1098/rspa.1998.0164}
}

@article{Griffiths1996,
  author  = {Griffiths, Robert B. and Niu, Chi-Sheng},
  title   = {Semiclassical {Fourier} transform for quantum computation},
  journal = {Phys. Rev. Lett.},
  volume  = {76},
  pages   = {3228},
  year    = {1996},
  doi     = {10.1103/PhysRevLett.76.3228}
}

@book{Kitaev2002,
  author    = {Kitaev, A. Yu. and Shen, A. H. and Vyalyi, M. N.},
  title     = {Classical and Quantum Computation},
  series    = {Graduate Studies in Mathematics},
  volume    = {47},
  publisher = {American Mathematical Society},
  address   = {Providence, RI},
  year      = {2002},
  doi       = {10.1090/gsm/047},
  url       = {https://doi.org/10.1090/gsm/047}
}

@article{Feynman1986,
  author  = {Feynman, Richard P.},
  title   = {Quantum mechanical computers},
  journal = {Found. Phys.},
  volume  = {16},
  pages   = {507},
  year    = {1986},
  doi     = {10.1007/BF01886518}
}

@article{Bloch2012,
  author  = {Bloch, Immanuel and Dalibard, Jean and Nascimb{\`e}ne, Sylvain},
  title   = {Quantum simulations with ultracold quantum gases},
  journal = {Nat. Phys.},
  volume  = {8},
  pages   = {267},
  year    = {2012},
  doi     = {10.1038/nphys2259}
}

@article{Monroe2021,
  author  = {Monroe, C. and Campbell, W. C. and Duan, L.-M. and Gong, Z.-X. and Gorshkov, A. V. and Hess, P. W. and Islam, R. and Kim, K. and Linke, N. M. and Pagano, G. and Richerme, P. and Senko, C. and Yao, N. Y.},
  title   = {Programmable quantum simulations of spin systems with trapped ions},
  journal = {Rev. Mod. Phys.},
  volume  = {93},
  pages   = {025001},
  year    = {2021},
  doi     = {10.1103/RevModPhys.93.025001}
}

@article{Browaeys2020,
  author  = {Browaeys, Antoine and Lahaye, Thierry},
  title   = {Many-body physics with individually controlled {Rydberg} atoms},
  journal = {Nat. Phys.},
  volume  = {16},
  pages   = {132},
  year    = {2020},
  doi     = {10.1038/s41567-019-0733-z}
}

@article{Childs2018,
  author  = {Childs, Andrew M. and Maslov, Dmitri and Nam, Yunseong and Ross, Neil J. and Su, Yuan},
  title   = {Toward the first quantum simulation with quantum speedup},
  journal = {Proc. Natl. Acad. Sci. U.S.A.},
  volume  = {115},
  pages   = {9456},
  year    = {2018},
  doi     = {10.1073/pnas.1801723115}
}

@article{Brassard2002,
  author        = {Brassard, Gilles and H{\o}yer, Peter and Mosca, Michele and Tapp, Alain},
  title         = {Quantum amplitude amplification and estimation},
  journal       = {Contemp. Math.},
  volume        = {305},
  pages         = {53--74},
  year          = {2002},
  doi           = {10.1090/conm/305/05215},
  url           = {https://doi.org/10.1090/conm/305/05215}
}

@article{Grinko2021,
  author  = {Grinko, Dmitry and Gacon, Julien and Zoufal, Christa and Woerner, Stefan},
  title   = {Iterative quantum amplitude estimation},
  journal = {npj Quantum Inf.},
  volume  = {7},
  pages   = {52},
  year    = {2021},
  doi     = {10.1038/s41534-021-00379-1}
}

@article{Huggins2022,
  author  = {Huggins, William J. and Wan, Kianna and McClean, Jarrod and O'Brien, Thomas E. and Wiebe, Nathan and Babbush, Ryan},
  title   = {Nearly optimal quantum algorithm for estimating multiple expectation values},
  journal = {Phys. Rev. Lett.},
  volume  = {129},
  pages   = {240501},
  year    = {2022},
  doi     = {10.1103/PhysRevLett.129.240501}
}

@article{Dutta2025,
title = {Exact space-depth trade-offs in multicontrolled Toffoli decomposition},
author = {Dutta, Suman and Wang, Siyi and Baksi, Anubhab and Chattopadhyay, Anupam and Maitra, Subhamoy},
journal = {Phys. Rev. A},
volume = {111},
issue = {5},
pages = {052611},
numpages = {12},
year = {2025},
month = {May},
publisher = {American Physical Society},
doi = {10.1103/PhysRevA.111.052611},
url = {https://link.aps.org/doi/10.1103/PhysRevA.111.052611}
}

@article{Babbush2018,
  author  = {Babbush, Ryan and Gidney, Craig and Berry, Dominic W. and Wiebe, Nathan and McClean, Jarrod and Paler, Alexandru and Fowler, Austin and Neven, Hartmut},
  title   = {Encoding electronic spectra in quantum circuits with linear {T} complexity},
  journal = {Phys. Rev. X},
  volume  = {8},
  pages   = {041015},
  year    = {2018},
  doi     = {10.1103/PhysRevX.8.041015}
}

@article{Baez2020,
author = {Maria Laura Baez  and Marcel Goihl  and Jonas Haferkamp  and Juani Bermejo-Vega  and Marek Gluza  and Jens Eisert },
title = {Dynamical structure factors of dynamical quantum simulators},
journal = {Proc. Natl. Acad. Sci. U.S.A.},
volume = {117},
number = {42},
pages = {26123-26134},
year = {2020},
doi = {10.1073/pnas.2006103117},
URL = {https://www.pnas.org/doi/abs/10.1073/pnas.2006103117}
}

@article{Piccinelli2026circuit,
  doi = {10.22331/q-2026-04-13-2060},
  url = {https://doi.org/10.22331/q-2026-04-13-2060},
  title = {A circuit-differentiation framework for {G}reen's functions on quantum computers},
  author = {Piccinelli, Samuele and Tacchino, Francesco and Tavernelli, Ivano and Carleo, Giuseppe},
  journal = {{Quantum}},
  issn = {2521-327X},
  publisher = {{Verein zur F{\"{o}}rderung des Open Access Publizierens in den Quantenwissenschaften}},
  volume = {10},
  pages = {2060},
  month = apr,
  year = {2026}
}

@article{Sels2021,
  title = {Quantum generative model for sampling many-body spectral functions},
  author = {Sels, Dries and Demler, Eugene},
  journal = {Phys. Rev. B},
  volume = {103},
  issue = {1},
  pages = {014301},
  numpages = {6},
  year = {2021},
  month = {Jan},
  publisher = {American Physical Society},
  doi = {10.1103/PhysRevB.103.014301},
  url = {https://link.aps.org/doi/10.1103/PhysRevB.103.014301}
}
\bibliographystyle{apsrev4-2}

\clearpage
\onecolumngrid
\appendix

\renewcommand\thefigure{S\arabic{figure}}  
\setcounter{figure}{0}

\end{document}